\documentclass[11pt]{article} 

\RequirePackage[colorlinks,citecolor=blue,linkcolor=blue,urlcolor=blue]{hyperref}

\usepackage{tikz}
\usepackage{amsmath}
\usepackage{amsthm}
\usepackage{amssymb}
\usepackage[a4paper, total={6.5 in, 9.5 in}]{geometry}
\usepackage{subcaption}
\usepackage{bbm}
\usepackage{mathtools}
\usepackage{setspace}
\let\origfootnote\footnote
\renewcommand{\footnote}[1]{%
    \begingroup
    \setstretch{1}%
    \origfootnote{#1}%
    \endgroup
}

\usepackage{colonequals}
\usepackage{natbib}
\usepackage{etoolbox}
\AtBeginEnvironment{thebibliography}{
    \footnotesize
    \singlespacing
}
\setcitestyle{open={(},close={)}}

\usepackage{proof-at-the-end}
\usepackage{authblk}
\usepackage[shortlabels]{enumitem}
\usepackage{booktabs}

\usepackage{xcolor}

\usetikzlibrary{calc,intersections}
\usetikzlibrary{arrows,calc}

\theoremstyle{plain}

\newtheorem{assumption}{Assumption}

\newtheorem{corollary}{Corollary}
\newtheorem{lemma}{Lemma}

\newtheorem{thm}{Theorem}

\newtheorem*{proposition*}{Proposition}
\newtheorem*{lemma*}{Lemma}

\theoremstyle{remark}

\NewDocumentEnvironment{thmrep}{O{}O{}+b}{ \begin{theoremEnd}[no link to proof,proof at the end,restate,#2]{thm}[#1]
#3
\end{theoremEnd}
}{}

\NewDocumentEnvironment{lemmarep}{O{}O{}+b}{ \begin{theoremEnd}[no link to proof,proof at the end,restate,#2]{lemma}[#1]
#3
\end{theoremEnd}
}{}

\NewDocumentEnvironment{propositionrep}{O{}O{}+b}{ \begin{theoremEnd}[no link to proof,proof at the end,restate,#2]{proposition}[#1]
#3
\end{theoremEnd}
}{}

\NewDocumentEnvironment{corollaryrep}{O{}O{}+b}{ \begin{theoremEnd}[no link to proof,proof at the end,restate,#2]{corollary}[#1]
#3
\end{theoremEnd}
}{}

\NewDocumentEnvironment{proof_inline}{b}{ 
\begin{proofEnd}
#1
\end{proofEnd}
}{}

\newcommand{\real}{\ensuremath{\mathbb R}}

\newcommand{\natnumb}{\ensuremath{\mathbb N}}

\newcommand{\bp}{\mathbf{p} }
\newcommand{\by}{\mathbf{y} }
\newcommand{\bmin}{\boldsymbol{\min} }

\newcommand{\C}{\mathcal C}

\newcommand{\Beta}{\mathrm{B}}

\title{
Identifying the Distribution of Welfare from Discrete Choice\protect\footnotetext{E-mail: \href{mailto:bart.capeau@ulb.be}{\texttt{bart.capeau@ulb.be}}, \href{mailto:sebastiaan.maes@uantwerpen.be}{\texttt{sebastiaan.maes@uantwerpen.be}}
\newline \textbf{Acknowledgements:} We are deeply grateful to Andr\'e Decoster for countless stimulating discussions and his unwavering encouragement. We also thank two anonymous referees, Debopam Bhattacharya, Laurens Cherchye, Beno\^it Decerf, Bram De Rock, Geert Dhaene, Peter Haan, Karim Kilani, Erwin Ooghe, Erik Schokkaert, and Frederic Vermeulen, as well as conference and workshop participants at KU Leuven, IAAE, IIPF, LAGV, LISER, and NESG for valuable suggestions and comments on earlier versions of this paper. Bart Cap\'eau acknowledges support from the FWO and the F.R.S.-FNRS (project number EOS~30544469). Bart Cap\'eau and Liebrecht De Sadeleer acknowledge support from the FWO (project number~G073020N). Sebastiaan Maes benefited from doctoral and postdoctoral fellowships of the FWO (project numbers 11F8919N and 12C8623N). The results and their interpretation are the authors' sole responsibilities.}
}

\author[$\dagger$]{Bart Cap\'eau
}
\author[$\mbox{}$]{Liebrecht De Sadeleer
}
\author[$\diamond$]{Sebastiaan Maes}

\affil[$\dagger$]{
Department of Economics, KU Leuven and ECARES, ULB
}

\affil[$\diamond$]{
Department of Economics, University of Antwerp
}

\date{\today}

\begin{document}

\maketitle
\begin{abstract}
\noindent Empirical welfare analyses often impose stringent parametric assumptions on individuals’ preferences and neglect unobserved preference heterogeneity. We develop a framework to conduct individual and social welfare analysis for discrete choice that does not suffer from these drawbacks. We first adapt the class of individual welfare measures introduced by \cite{Fleurbaey2009} to settings where individual choice is discrete. Allowing for unrestricted, unobserved preference heterogeneity, these measures become random variables. We then demonstrate that their distribution can be derived from choice probabilities, which can be estimated nonparametrically from cross-sectional data. Additionally, we derive nonparametric results for the joint distribution of welfare and welfare differences, and 
for social welfare. The former is an important tool in determining whether the winners of a price change belong disproportionately to those groups who were initially well-off.
\end{abstract}

		


\clearpage
\section{Introduction}\label{sec:introduction}

Discrete choice random utility models (DC-RUMs) have a long tradition in both theoretical and empirical microeconometric research. They have been applied to a wide range of problems in education, health care, industrial organization, labor, marketing, public finance, and transportation.\footnote{Some parametric models within this class, such as the binary and multinomial logit models, yield convenient closed-form choice probabilities, which makes them a popular choice in applied work. For a comprehensive overview, see \cite{trainDiscreteChoiceMethods2003}.} The success of DC-RUMs can be explained by their ability to model individual demand among a discrete set of alternatives in a flexible way, allowing for the presence of unobserved heterogeneity in individual preferences. As econometric models typically explain only a small part of the variation in choice data, unobserved heterogeneity is thought to be an important driver of individual demand in empirical applications. Neglect or misspecification of this heterogeneity might introduce substantial biases into the analysis.

We develop a framework to conduct individual and social welfare analysis in DC-RUMs that allows for unrestricted, unobserved heterogeneity in individuals' preferences. Our revealed preference approach is entirely nonparametric and, therefore, does not suffer from misspecification in the econometric model. The framework is sufficiently general to study both \emph{levels} and \emph{differences} of individual welfare, where the latter measure individuals' gains or losses induced by an exogenous price change. Characterizing these concepts is of first-order importance to applied welfare analysis for at least three reasons. Firstly, knowledge of levels of welfare enables researchers to rank individuals according to their well-being in any given situation, distinguishing between those who are well-off and those who are less well-off. In aggregating these levels across individuals, overall social welfare can be calculated and compared between two situations. Secondly, knowledge on differences of welfare allows to assess individuals' welfare gains or losses from a price change, distinguishing between winners and losers. Thirdly, joint knowledge on levels and differences of welfare reveals the association between individuals' gains or losses from a price change and their position in terms of initial welfare. This allows for the assessment of, for example, whether the winners of a price change belong disproportionately to those groups who were initially well-off.

Our results complement and extend the recent findings of \citeauthor{bhattacharyaNonparametricWelfareAnalysis2015} (\citeyear{bhattacharyaNonparametricWelfareAnalysis2015}, \citeyear{bhattacharyaEmpiricalWelfareAnalysis2018}), who studies the distribution of the compensating and equivalent variation (CV and EV), in several important directions.\footnote{This is of theoretical and practical interest. For example, an important limitation of the CV and EV is that they cannot be used in the labor supply context, as both the initial and final prices (i.e., wages) differ across individuals. It has been shown that not using a common reference price in welfare analysis exhibits unattractive features \citep{king,CapeauDecosterDeSadeleer2023}.} First, we do not only consider welfare differences, but also derive nonparametric results for the distribution of welfare levels and for the joint distribution of welfare levels and differences. In doing so, this paper is the first to study welfare levels in a nonparametric setting with unrestricted, unobserved heterogeneity. Second, we go beyond \citeauthor{Samuelson1974}'s (\citeyear{Samuelson1974}) money metric utilities (MMUs), and show that our results hold for a much broader class of welfare metrics. Our results characterize what can be learned about individual and social welfare from cross-sectional and panel data. Third, we also provide all results conditional on the observed pre- or post-price change choices, which can reduce the uncertainty in the welfare estimates. Moreover, conditioning on these observed choices might also be important from a political economy perspective.

To operationalize our framework, we first adapt the class of individual welfare measures introduced by \cite{Fleurbaey2009} to settings where individual choice is discrete instead of continuous. We call them nested opportunity set~(NOS) measures.\footnote{In a continuous choice setting, these measures coincide with the class of individual welfare measures which follow the \emph{equivalence approach} discussed and advocated by \cite{Fleurbaey2009,Fleurbaey2011}. In the context of ranking distributions of bundles of goods across individuals, \cite{BosmansDecancqOoghe2018} use the NOS measures in their  \emph{reference set welfarism} criterion for social rankings. They are also used by \cite{piacquadio17} in his fairness characterization of utilitarianism.} These welfare measures were developed in line with the growing consensus that well-being, be it measured at the individual micro- or at a nation-wide macro-level, needs to be assessed in a multi-dimensional way, which goes beyond income alone (e.g., see \citealp{StiglitzSenFitoussi} and \citealp{fleurbaey2013beyond}). 

We show that MMUs are within the class of NOS measures, and use them as a leading example to illustrate our approach.\footnote{The use of welfare measures based on the expenditure function, the so-called MMUs, is a well-established practice in the applied welfare literature (for seminal contributions, see \citealp{diamondmcfadden}; \citealp{dixit}; \citealp{king}).} As a result, the well-known CV and EV, which are both measures of differences of individual welfare, are embedded in our framework. Our results therefore generalize the findings of
\cite{johnk.dagsvikCompensatingVariationHicksian2005} and \cite{depalmaTransitionChoiceProbabilities2011} to settings where unobserved heterogeneity is essentially unrestricted. 

As the presence of unobserved heterogeneity renders these NOS measures stochastic from the point of view of the econometrician, we then show how their distributions relate to what is typically observed in cross-sectional and panel data. In particular, we  prove that the marginal distribution of NOS measures can be recovered nonparametrically from cross-sectional data by evaluating the observed choice probabilities at counterfactual prices. This allows researchers to study levels of individual welfare in any given situation. Likewise, we show that the joint distribution of welfare levels and welfare differences can be recovered nonparametrically from panel data by evaluating the observed transition probabilities at counterfactual prices.\footnote{These transition probabilities are derived under the assumption that unobserved individual preferences are unaltered by the price change, which implies perfect correlation in unobserved heterogeneity before and after the price change. Alternatively, \cite{johnk.dagsvikDiscreteChoiceContinuous2002} and \cite{dellesiteTransitionChoiceProbabilities2013} consider models where there is imperfect correlation.} Building on these results, we are able to nonparametrically characterize levels and differences in aggregate welfare for any additively separable social welfare function.

Our identification  results are constructive and can be implemented in empirical work using only nonparametric regression. We also demonstrate how Boole-Fr\'echet inequalities \citep{frechet} and stochastic revealed preference restrictions can be exploited to construct sharp bounds on the transition probabilities in the common event when only cross-sectional data is available.\footnote{In the context of continuous choice, similar inequalities have been exploited by \cite{hoderleinstoye}, \cite{kitanurastoye}, and \cite{debkitamura22}.} These bounds are functionals of the choice probabilities and are, as such, straightforward to implement. They can readily be used to set-identify the concepts that are expressed in terms of transition probabilities.

As a by-product, we can not only condition our results on exogenous characteristics, but also derive results conditional on the choices before or after the price change. This allows researchers to take the additional information conveyed by the observed choices into account. Conditioning on observed choices restricts the admissible set of unobserved preference heterogeneity, such that individual welfare can be measured more precisely. In addition, this conditioning is also relevant from a political economy perspective. It provides answers to questions like ‘What is the welfare impact of a refundable tax credit on the \emph{unemployed}?’ and ‘How do congestion taxes affect the welfare of \emph{drivers}?’. Finally, in the labor supply setting, these conditional results are also required to study how welfare gains vary with the level of earned income, as the latter is endogenous.

\paragraph{Related literature.}
The structural modeling of individual preferences in DC-RUMs renders this class of models especially suitable for the welfare analysis of price changes. Over the last fifteen years, a methodological literature has emerged that derives closed-form expressions for the distribution of the~CV and~EV under ever less parametric assumptions on the nature of individuals' preferences.\footnote{Before, no closed-form expressions existed, even for the expected values of the CV and EV. Therefore, researchers had to resort to approximations, except for the most simple of DC-RUMs in which individuals have constant marginal utility of income and unobserved heterogeneity is additive and generalized extreme value distributed  (\citeauthor{smallAppliedWelfareEconomics1981}, \citeyear{smallAppliedWelfareEconomics1981};  \citeauthor{mcfaddenComputingWillingnesstoPayRandom1999}, \citeyear{mcfaddenComputingWillingnesstoPayRandom1999}). These approximations are either biased \citep{moreyRepeatedNestedLogitModel1993}, rather uninformative \citep{herrigesNonlinearIncomeEffects1999}, or computationally burdensome \citep{mcfaddenComputingWillingnesstoPayRandom1999}.}
For the class of additive DC-RUMs, \cite{johnk.dagsvikCompensatingVariationHicksian2005} provide expressions for the distribution of the CV based on compensated (Hicksian) choice probabilities. The authors provide 
analytical results 
for models where unobserved heterogeneity is generalized extreme value distributed. Alternatively, \cite{depalmaTransitionChoiceProbabilities2011} advance a direct approach for this class, in which they express this distribution in terms of uncompensated (Marshallian) choice probabilities. More recently, \citeauthor{bhattacharyaNonparametricWelfareAnalysis2015} (\citeyear{bhattacharyaNonparametricWelfareAnalysis2015}, \citeyear{bhattacharyaEmpiricalWelfareAnalysis2018}) showed that the marginal distributions of the CV and EV can be written as a 
functional of uncompensated choice probabilities, even when unobserved heterogeneity is essentially unrestricted, and therefore possibly nonadditive.
This paper generalizes and extends these results as described above.

Several semiparametric methods have been developed to relax functional form assumptions on either deterministic preferences or the distribution of unobserved heterogeneity in DC-RUMs (for early results see \citealp{manskiMaximumScoreEstimation1975};  \citealp{matzkinSemiparametricEstimationMonotone1991}; and \citealp{kleinEfficientSemiparametricEstimator1993}). Other contributions introduce entirely nonparametric methods that do not impose functional form restrictions on either of these components for this class of models, based on either shape restrictions (e.g., see \citealp{matzkinNonparametricIdentificationEstimation1993}) or the availability of a large-support {special regressor} (e.g., see \citealp{lewbelSemiparametricQualitativeResponse2000} and \citealp{brieschNonparametricDiscreteChoice2010}). The approach we follow in this paper deviates from this literature as our objective is not to recover deterministic preferences and the distribution of unobserved heterogeneity, but instead to identify individual welfare measures which are functions of both these model primitives.

In recent years, a comprehensive theoretical framework for measuring individual well-being has been developed that encompasses both the classical MMUs~(\citealp{Samuelson1974}), adaptations of other measures like~\citeauthor{Pazner1979}'s (\citeyear{Pazner1979}) ray utilities, and measures like the equivalent income and wage metrics (among others, see \citealp{pencavel1977constant}; \citealp{fleurbaey2007two}; \citealp{Fleurbaey2009}; 
\citealp{decancq2015happiness}; and \citealp{Fleurbaey2017}).
Almost all of these measures cardinalize preferences by associating their indifference sets with members of a  family of nested opportunity sets; i.e., a lower ranked indifference set is associated with a smaller opportunity set. The sizes of those opportunity sets are argued to be an ethically more meaningful basis for interpersonal comparisons of well-being than income or reports on subjective satisfaction levels. Indeed, contrary to income and subjective satisfaction, such measures ensure that individuals with the same preferences and in a situation which makes them indifferent among each other are always considered to be equally well-off. We adapt this class of individual welfare measures to settings where individual choice is discrete and show how the distribution of these NOS measures relates to what is typically nonparametrically observed in cross-sectional and panel data. 

Another strand of literature focuses on the nonparametric identification of counterfactual choices and welfare under unobserved heterogeneity in models where demand is continuous instead of discrete. Most results exploit the smoothness of the underlying individual demand functions to arrive at Slutsky-like restrictions on average and quantile demands (e.g., see \citealp{detteTestingMultivariateEconomic2016}; \citealp{hausmanIndividualHeterogeneityAverage2016a}; \citealp{blundellNonparametricEstimationNonseparable2017}; and \citealp{hoderleinEstimatingDistributionWelfare2018}). Other results exploit the axioms of revealed preference to attain identification under the presence of unobserved heterogeneity  (e.g., see \citealp{blundellBoundingQuantileDemand2014}; and \citealp{cosaertNonparametricWelfareDemand2018}). In contrast to our results, however, the availability of cross-sectional and short panel data is generally not sufficient to point-identify the distribution of welfare levels and differences in settings where demand is continuous and unobserved heterogeneity is unrestricted.

Finally, this paper contributes to the literature that applies NOS measures empirically. Using microdata, \cite{Bargain_2013} and \cite{Decoster2015} estimate  parametric DC-RUMs to model labor supply and construct rankings of households based on NOS measures. \cite{Carpantier2016} extend the approach of \cite{Decoster2015} by integrating unobserved preference heterogeneity through a numerical procedure (comparable to the approach of \cite{herrigesNonlinearIncomeEffects1999} for welfare differences). Our results show that the parametric assumptions imposed in these papers are not necessary to obtain identification.

\paragraph{Outline of the paper.}

The remainder of this paper is organized as follows. In the next section, we introduce the class of NOS welfare measures for settings where choice is continuous. In Section \ref{sec:formal_definitions}, our conceptual framework is laid out. We specify the DC-RUM and adapt the class of NOS welfare measures to this discrete setting. In Section \ref{sec:distribution}, we show that the distribution of these objects can be derived from choice and transition probabilities. In addition, we derive nonparametric results for the joint distribution of welfare and welfare differences, as well as for social welfare. In Section~\ref{sec:implement}, we discuss how the choice and transition probabilities can be retrieved from cross-sectional data. Section~\ref{sec:conclusion} contains concluding remarks. All proofs and some additional results are collected in the Appendix.

\section{NOS welfare measures in a continuous setting}\label{sec:metrics_continuous}
In this section, we briefly explain and motivate the class of NOS welfare measures, which have been introduced by~\cite{Fleurbaey2009} for settings where choice is continuous.\footnote{The discussion here is informal as additional assumptions are needed to guarantee the existence and the uniqueness of the NOS welfare measure. A more rigorous treatment is postponed to Section \ref{ssec:metric_discrete}.} These measures cardinalize preferences by associating each indifference set with a member of a family of nested opportunity sets, which is common to all individuals.\footnote{Note that these opportunity sets are a conceptual device to cardinalize preferences and are unrelated to individuals' actual budget sets.} A lower ranked indifference set is associated with a smaller set from that family, such that the size of the opportunity set acts as a measure of individual well-being, respecting that individual's preferences. 

Formally, let $\Beta \subseteq \mathbb{R}^n_+$ be the set of all bundles~$\mathbf{b}$ an agent can possibly obtain, and let $\{B_\lambda \subseteq \Beta  \mid \lambda \in \Lambda \subseteq \mathbb{R}\}$ denote a family of nested opportunity sets indexed with a parameter $\lambda$ such that $\lambda < \lambda'$ implies that $B_{\lambda} \subsetneq B_{\lambda'}$. Given a well-behaved utility function $U(\mathbf{b}): \Beta \rightarrow \mathbb{R}$, the NOS welfare measure evaluated in a bundle $\mathbf{b}\in \Beta$ is then defined as 
 \begin{equation}
    \label{eq:metriccont}
     W(\mathbf{b}) = \max\left\{\lambda \in \Lambda \mid U(\mathbf{b}) \geq \max_{\mathbf{b}' \in B_\lambda}U(\mathbf{b}')\right\},
 \end{equation}
 that is, the largest value of $\lambda$ --- or, equivalently, the largest opportunity set --- for which the individual still weakly prefers bundle $\mathbf{b}$ above all bundles $\mathbf{b}'$ in the set $B_\lambda$.

\begin{figure}
\begin{subfigure}{0.5\textwidth}
    \centering
\begin{tikzpicture}[scale=2.25,domain=0:1]
    \draw [<->,thick] (0,2.5) node (yaxis) [above] {good 2}
        |- (2.7,0) node (xaxis) [right] {good 1};
 \def\a{0.388658088} 
  \def\b{0.35+\a}
 \def\c{0.665} 
 \def\d{atan(\c)}
        \path
    coordinate (fullhealthstart) at (2.5,0)
        coordinate (fullhealthtop) at (2.5,2.5)
    coordinate (c1) at (0,2.25)
    coordinate (c2) at (0,\b)
    coordinate (c3) at (2.5,1.5)
    coordinate (c4) at (1,\a)
    coordinate (slut) at (2.7,.5)
    coordinate (top) at (4.2,2);

       \filldraw [color=black,thick,fill=darkgray!20,semitransparent]  (0,2) parabola[bend pos = 0]  (2.5,0.5)--(2.5,0)--(0,0)--cycle;

  \filldraw [color=black,thick,fill=darkgray!70,semitransparent]  (0,1) parabola[bend pos = 0] (2.5,0)--(2.5,0)--(0,0)--cycle;

        \draw[name path = R1,color = black,thin]  (c1) parabola[bend pos = 1] (c3) ;

\draw[name path = R3,color = black,thick]  (0,1.5) parabola[bend pos = 1] (2.5,0.505) ; 
                
               \node[ black,  below left ] at (1.2,0.7) {$B_{\lambda_1}$} ;
%
%
                       \node[ black,  below left ] at (1.4,1.5) {$B_{\lambda_2}$} ;
   \fill[black] (2,1.53) circle (0.5pt) node[below right , black] {$\mathbf{b}_2$};
    \fill[black] (2.3,0.514) circle (0.5pt) node[below, black] {$\mathbf{b}_1$};


\end{tikzpicture}
\caption{NOS as a measure of well-being}
    \label{fig:noswb}   
\end{subfigure}%
\begin{subfigure}{0.5\textwidth}
  \centering
\begin{tikzpicture}[scale=2.25,domain=0:1]
   \draw [<->,thick] (0,2.5) node (yaxis) [above] {good 2}
        |- (2.7,0) node (xaxis) [right] {good 1};
 \def\a{0.388658088} 
  \def\b{0.35+\a}
 \def\c{0.665} 
 \def\d{atan(\c)}
        \path
    coordinate (fullhealthstart) at (2.5,0)
        coordinate (fullhealthtop) at (2.5,2.5)
        coordinate (c1) at (0,2.05)
        coordinate (c1prime) at (0,1.60)
    coordinate (c2) at (0,\b)
    coordinate (c3) at (2.5,0.45)
    coordinate (c3prime) at (2.5,0.75)
    coordinate (c4) at (1,\a)
    coordinate (slut) at (2.7,.5)
    coordinate (top) at (4.2,2);
    
    \draw[name path = R1,color=black,thick]  (c1) parabola[bend pos = 1] (c3) ;
    \node[ black,  above right ] at (c1) {$R$} ;
    \draw[name path = R2,color=black,thick]  (c1prime) parabola[bend pos = 1] (c3prime)  ;
        \node[ black,  above right ] at (c1prime) {$R'$} ;
 

     \def\e{0.91}
     \def\f{1.17}
     \draw [gray,dashed,thick,name path = refwage1](0,\e)--(2.5,\e-1/2);
     \draw [gray,dashed,thick,name path = refwage2](0,\f)--(2.5,\f-1/2);
%



   \fill[black] (0.915,1.09) circle (0.5pt) node[above right , black] {$\mathbf{b}$};
     \fill[black] (1.79,0.815) circle (0.5pt) ;
     \fill[black] (2.1,0.485) circle (0.5pt) ;     
\end{tikzpicture}
        \caption{The MMU measure}
    \label{fig:MMU}
\end{subfigure}
\caption{NOS welfare measures in a continuous choice setting}
\end{figure}
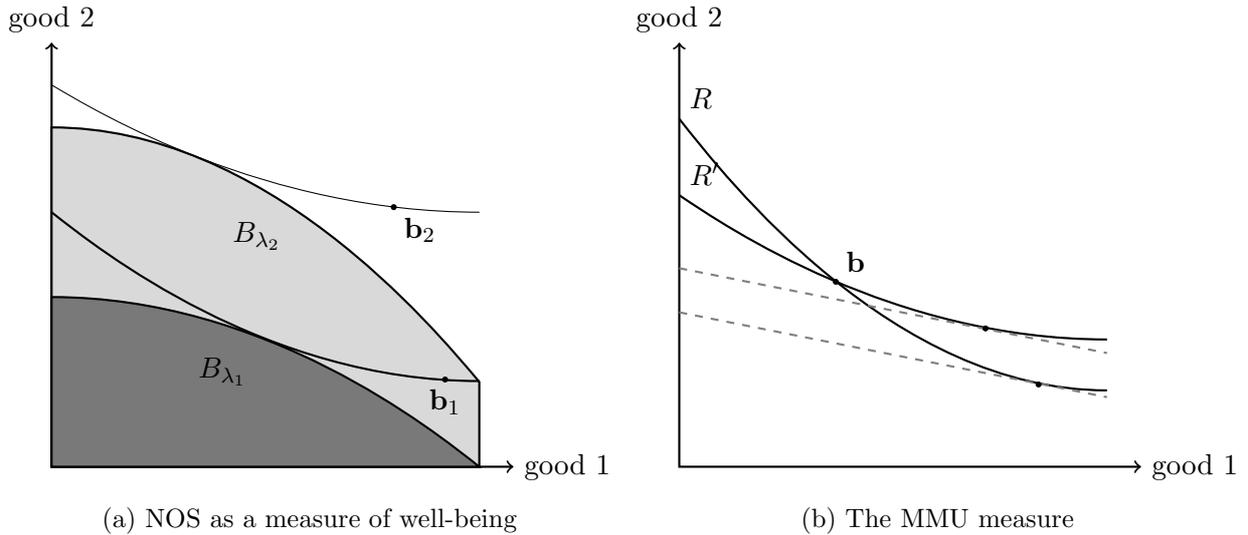

 This definition is illustrated in Figure~\ref{fig:noswb} by means of a classical trade-off between two goods. Suppose first an individual obtains a bundle~$\mathbf{b}_1$, and let the thick black line denote her indifference curve passing through~$\mathbf{b}_1$.
 This indifference curve is associated with the opportunity set $B_{\lambda_1}$, which is shaded in dark gray. In accordance with the definition in Equation~\eqref{eq:metriccont}, this set is designed such that the individual could obtain, at best, a bundle equally good as~$\mathbf{b}_1$, when she would be faced with the opportunity set~$B_{\lambda_1}$. Suppose now that the individual  obtains a bundle~$\mathbf{b}_2$, which is better than~$\mathbf{b}_1$, according to her own preferences. Then the associated opportunity set~$B_{\lambda_2}$, which is shaded in light gray, is again such that the best bundle in this set is equally good as~$\mathbf{b}_2$, and includes the  set $B_{\lambda_1}$. From this illustration, it is clear that the size of these opportunity sets serves as a measure of individual well-being that respects preferences, in the sense that the well-being level of an individual in situation~$\mathbf{b}_2$ is higher than her well-being level in situation~$\mathbf{b}_1$, if and only if that individual prefers~$\mathbf{b}_2$ to~$\mathbf{b}_1$. The size of an opportunity set is measured by its indexing parameter~$\lambda$. 

A well-known and often used set of welfare measures within the NOS class is the set of MMUs  \citep{Samuelson1974}.\footnote{ Other examples of NOS measures are the ray utilities of~\cite{Pazner1979} and the equivalent income metrics  introduced in~\cite{decancq2015happiness}.
}  In this case, the NOS are of the form
\begin{equation}
    B_{\lambda}\equiv\left\{\mathbf{b} \in \Beta \left| \mathbf{b}'\bp^{ref} \leq\lambda\right.\right\},
\end{equation}
where $\bp^{ref}$ is a vector of reference prices that is fixed by the researcher. In this specification, the shape of the family of nested opportunity sets is  determined by a set of hyperplanes with normal $\mathbf{p}^{ref}$. The vector of reference prices determines the slope of the hyperplane. 

Applying Equation~\eqref{eq:metriccont}, we find that the MMU measures well-being in a bundle by the maximal monetary amount that can be granted to an individual faced with reference prices~$\mathbf{p}^{ref}$, such that she would at most be equally well-off as in that bundle. This coincides with the expenditure function representation of preferences. The indexing parameter $\lambda$ equals $\mathbf{b}'\bp^{ref}$ and can thus be interpreted as a monetary amount.
In Figure~\ref{fig:MMU}, we illustrate the opportunity sets associated with these measures. The dashed lines  define the upper bound of two elements from the family of this particular MMU with reference prices $\left(p^{ref}_1,p^{ref}_2\right)$ and have slope $-p_1^{ref}/p_2^{ref}$. For a particular bundle~$\mathbf{b}$, the person with indifference curve $R'$ through that bundle is considered to be better off than the person with indifference curve $R$.

\begin{figure} 
  \centering
\begin{tikzpicture}[scale=3.00,domain=0:1]
   \draw [<->,thick] (0,2.75) node (yaxis) [above] {good 2}
        |- (4.45,0) node (xaxis) [right] {good 1};

 \def\a{0.388658088} 
  \def\b{0.35+\a}
 \def\c{0.665} 
 \def\d{atan(\c)}
        \path

    coordinate (c1) at (0.5,2)
    coordinate (c1prime) at (3.1,0.6)
    coordinate (c2) at (0,\b)
    coordinate (c3) at (3,0.25)
    coordinate (c3prime) at (2.5,0.75)
    coordinate (c4) at (1^0.5,1^.5)
    coordinate (slut) at (2.7,.5)
    coordinate (top) at (4.2,2);

   \draw[domain=0.5:3, smooth, variable=\x, black, thick] plot ({\x}, {1/\x});
       \node[ black, above  right ] at (c1) {$R$} ;
   \draw[domain=0.25:3, smooth, variable=\x, black, thick] plot ({\x}, {(1/\x^0.35)^(1/.65)});
        \node[ black,  above left ] at (c1prime) {$R'$} ;

     \def\e{1.2}
     \def\f{2.71}
     \def\g{1.0325}
     \def\sa{4.4/\e}
     \def\h{2.355}
     \def\sb{1.45/\f}
     \draw [gray,thick,name path = refwage1](0,\e)--(4.4,0);
     \draw [gray,thick,name path = refwage1](0,\g)--(\g*\sa,0);
     \draw [gray,dashed,thick,name path = refwage2](0,\f)--(1.45,0);
      \draw [gray,dashed,thick,name path = refwage2](0,\h)--(\h*\sb,0);

   \fill[black] (1,1) circle (0.75pt) node[above right , black] {$\mathbf{b}$};
     \fill[black] (0.7+0.01,\f-0.7*\f/1.45+0.01) circle (0.75pt) ;
     \fill[black] (0.45+0.01,\h-0.45*\f/1.45+0.01) circle (0.75pt) ;
     \fill[black] (1.55,\e-1.55*\e/4.4+0.01) circle (0.75pt) ;  
     \fill[black] (2.0,\g-2.0*\e/4.4+0.01) circle (0.75pt) ;  

\end{tikzpicture}
        \caption{Shape of opportunity sets and interpersonal comparisons of well-being}\label{fig:shapeNOS}
\end{figure}
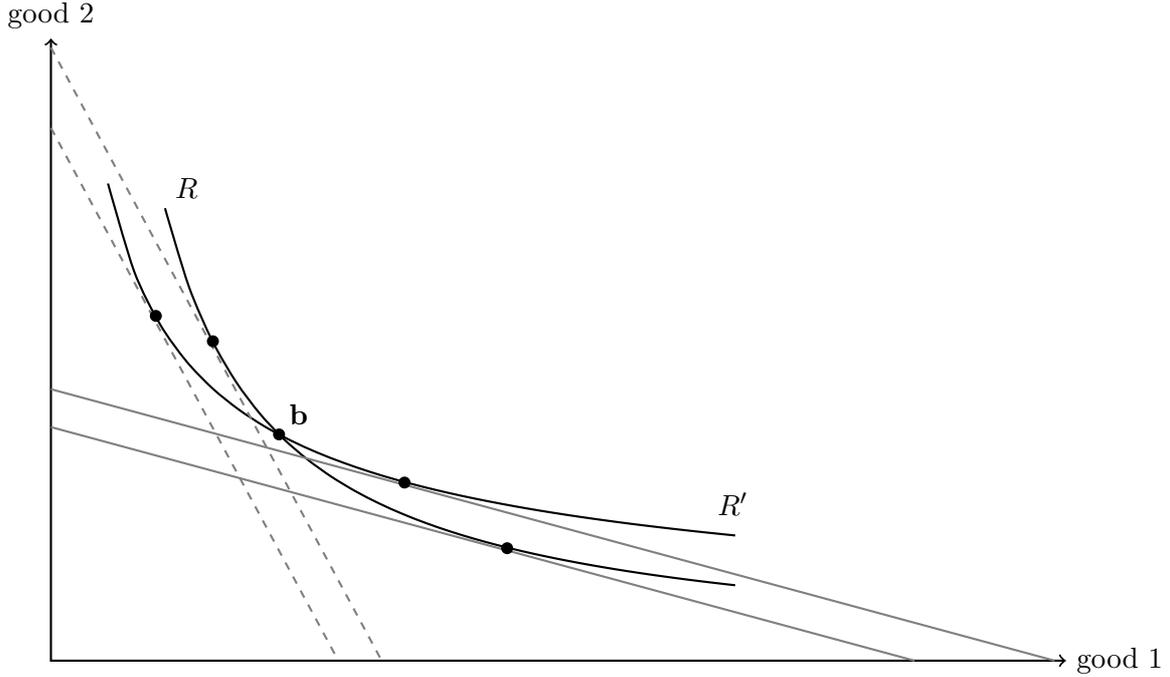
 
 The shape of the opportunity sets determines how interpersonal comparisons are made.  In Figure~\ref{fig:shapeNOS}, we illustrate this for two choices of reference prices within the class of MMUs. The full gray lines represent the upper boundaries of two opportunities sets of an MMU for a given choice  of reference prices. The dashed gray lines are those for another choice of reference prices. Again consider two persons, both endowed with a bundle $\mathbf{b}$. According to the first MMU, the individual with indifference curve $R$ through that bundle is considered to be worse off than the person with indifference curve $R'$. The reverse holds true for the second MMU. 
 
 Normative analyses should advocate principles for choosing a particular set of reference prices, or more generally, for choosing a particular NOS~measure.\footnote{For example, \cite{Fleurbaey2017} provide a characterization of the normative principles laying behind the choice between MMUs and ray utilities as measures of individual well-being.} By contrast, the approach and results developed in this paper are positive: we provide identification results for any welfare measure that is within the NOS class, but do not take a stance on which measure should be chosen in applied welfare analyses. This choice should be guided by researchers' ethical priors and the specific application at hand.

\section{Conceptual framework}\label{sec:formal_definitions}

Our conceptual framework allows for unrestricted, unobserved heterogeneity in DC-RUMs. As this set-up does not impose any restrictions on \textit{observed} individual characteristics, all results in this paper can be thought of as being conditional on these covariates.

\subsection{Discrete choice model}\label{ssec:DCM}

We summarize here the main building blocks and assumptions of DC-RUMs (for a detailed  technical review, see  \citealp{mcfadden1981econometric,mcfaddenRevealedStochasticPreference2005}).
\paragraph{Preferences.}
	Let $\mathcal{H}$ denote the universe of preference types and let $\Pr_\eta$ represent the distribution of these preference types in the population. Every preference type $\eta$ can be thought of as a different individual, who has idiosyncratic preferences over bundles $(y-p_c,c)$, composed of a numeraire good and elements of a finite (and common) set of alternatives $\mathcal{C}$, with $\left|\mathcal{C}\right|\colonequals n\in\natnumb \setminus \{0\}$. The numeraire good can be interpreted as income. The price $p_c$ is then the financial cost of choosing alternative~$c$. As such, the preferences can be thought of as conditional indirect preferences where the individual has in each option~$c$ a budget $y-p_c$ that she allocates optimally across continuously divisible  commodities available at exogenously given commodity prices~(\citealp{mcfadden1981econometric}). 
    The idiosyncratic preferences are assumed to be representable by a utility function $U^\eta_c(y - p_c)\colonequals U(y - p_c, c, \eta): \mathbb{R} \times \mathcal{C} \times  \mathcal{H}  \rightarrow \mathbb{R}$, in which $y$ denotes exogenous income and $p_c$ the price of alternative~$c \in \mathcal{C}$.
	Prices for all alternatives, $\left(p_c, c\in \mathcal{C}\right)$, are collected in a vector denoted by~$\mathbf{p}$ and we will call $\left(\mathbf{p},y\right)$ a budget set. Residual income in alternative $c$ is defined as the amount of the numeraire left over after choosing this alternative, i.e., $y-p_c$.
	
	Note that our formulation of preferences is very flexible as it allows them to differ arbitrarily across individuals.\footnote{Common parametric utility specifications in additive DC-RUMs (where $U^\eta_c(y - p_c) = V_c(y-p_c) + \varepsilon_c(\eta)$), or additive DC-RUMs models with random coefficients (where $U^\eta_c(y - p_c) = V_c(y-p_c, \beta(\eta)) + \varepsilon_c(\eta)$) are encompassed by our approach.  
} The only economically substantial restriction we will impose on this function is that utility is continuous and strictly increasing in the numeraire.
	\begin{assumption}
		\label{assumption1}
		Individual preferences are represented by a utility function $U^\eta_c(y - p_c)$ that is continuous and strictly increasing in the numeraire for every preference type $\eta \in \mathcal{H}$ and every alternative $c \in \mathcal{C}$. Moreover, preferences satisfy the following regularity conditions:
		(R1) For each pair of alternatives $c, c' \in \mathcal{C}$, and for each fixed $y$ and $p_c$, it holds that $U^\eta_c(y-p_c) > \lim_{p_{c'} \rightarrow \infty} U^\eta_{c'}(y-p_{c'})$ and that $U^\eta_c(y-p_c) < \lim_{p_{c'} \rightarrow - \infty} U^\eta_{c'}(y-p_{c'})$. (R2) For every budget set $(\mathbf{p}, y)$, the set of types that are indifferent between two or more alternatives in the choice set~$\mathcal{C}$ has probability measure zero.
	\end{assumption}
\noindent This assumption is ubiquitous in empirical work that employs (semi)parametric DC-RUMs. Monotonicity in the numeraire establishes the existence and uniqueness of our welfare measures and yields stochastic revealed preference conditions that we will exploit to obtain the identification results. Regularity condition (R1) ensures that when the price of a given alternative goes to infinity, it will never be preferred above another alternative with a finite price. Analogously, when the price of a given alternative goes to minus infinity --- or equivalently residual income in that alternative goes to plus infinity --- it will always be preferred above another alternative with a finite price. The negligibility of indifferences between alternatives (R2) ensures that no tie-breaking rule has to be established.

	In addition, we also assume that the distribution of the preference types, denoted by $F(\eta)$, is independent of the budget set $(\mathbf{p}, y)$.

	\begin{assumption}
		\label{assumption2}
		The distribution of unobserved heterogeneity $F(\eta)$ is independent of prices $\mathbf{p}$ and exogenous income $y$: i.e., $F(\eta \mid \mathbf{p}, y) = F(\eta)$.
	\end{assumption}
	\noindent The exogeneity of budget sets is a strong, but standard, assumption in the literature on nonparametric identification of individual demand and welfare (e.g., see \citealp{hausmanIndividualHeterogeneityAverage2016a}). Indeed, to the best of our knowledge, there are no theoretical results that allow for general forms of endogeneity in the presence of unrestricted, unobserved heterogeneity.  

	\paragraph{Individual choice behavior.}
	Finally, we assume that observed choice behavior is actually generated by a DC-RUM.
    This assumption entails that an individual $\eta$ chooses a given alternative $i$, if and only if this alternative yields the highest utility among the elements in her choice set $\mathcal{C}$.
	\begin{assumption}
		\label{assumption3}
		Let $J^\eta(\mathbf{p}, y) \equiv J(\mathbf{p}, y, \eta) : \mathbb{R}^{n+1} \times \mathcal{H} \rightarrow \mathcal{C}$ denote the individual demand function. It holds that $J^\eta(\mathbf{p}, y) = i \iff
		U^\eta_i(y - p_{i}) \geq \max_{c \neq i}\{ U^\eta_{c} (y - p_{c})                   \}$.
	\end{assumption}
\noindent 
    Note that individual demand is single-valued with probability one as one can neglect indifferences between alternatives by regularity condition (R2) in Assumption~\ref{assumption1}.

\paragraph{Choice and transition probabilities.}
	The individual choices induced by a DC-RUM are stochastic from the point of view of the econometrician, as the preferences types are non-observable. When this random variation is averaged out across types, one obtains a set $\{P_i(\mathbf{p}, y)\}_{i \in \mathcal{C}}$ of uncompensated (Marshallian) conditional choice probabilities,
	\begin{equation} \label{eq:cond_choice}
		\begin{split}
			P_i(\mathbf{p}, y) &= \Pr_\eta \left[\Big \{U^\eta_i(y - p_i) \geq \max_{c \neq i }\{U^\eta_{c}(y-p_{c})\}\Big\}\right] \\
			&= \Pr_\eta\left[J^\eta(\bp, y) = i\right] \\
			&= \int_\mathcal{H} \mathbbm{I} \left[ J^\eta(\bp, y) = i \right] d F(\eta \mid \mathbf{p},y) \\
			&= \int_\mathcal{H} \mathbbm{I} \left[ J^\eta(\mathbf{p}, y) = i \right] d F(\eta),\\
		\end{split}
	\end{equation}
	where $\mathbb{I}[\cdot]$ denotes the indicator function.\footnote{These choice probabilities are designated \emph{conditional} as they depend on a vector of prices and income. In the interest of brevity, this qualification will be dropped in the sequel.} The last expression asymptotically coincides with the observed choice frequency for every alternative $i \in \mathcal{C}$, conditional on the budget set $(\mathbf{p}, y)$.\footnote{This concept is also known as the \textit{average structural function} (e.g., see \citealp{blundellEndogeneitySemiparametricBinary2004}). The asymptotic equivalence follows from the law of large numbers as the choice probabilities are essentially conditional expectation functions.} If cross-sectional data contains enough relative price and exogenous income variation, these objects are nonparametrically estimable.\footnote{It is clear from Equation \eqref{eq:cond_choice} that these probabilities are composed of both the utility function $U^\eta_c$ and the distribution of unobserved heterogeneity $F$. As such, they are not sufficiently informative to separately identify these two model primitives. Fortunately, knowledge on such primitives is not necessary for our purposes.}

	Another concept induced by DC-RUMs is the set $\{P_{i, j}(\mathbf{p}, \mathbf{p}', y)\}_{i,j \in \mathcal{C}}$ of  uncompensated conditional {transition} probabilities. These probabilities are formally defined as
	\begin{equation}\label{eq:transprob}
		\begin{split}
			P_{i, j}(\mathbf{p}, \mathbf{p}', y)&= 
				\Pr_\eta \Big[ \Big \{U^\eta_i(y - p_i) \geq \max_{c \neq i}\{U^\eta_c(y-p_c)\}\Big\}  \\
				&\qquad\qquad
			\cap \Big \{U^\eta_{j}(y - p'_j) \geq \max_{c \neq j} \{U^\eta_c (y-p'_c)\}\Big\} \Big] \\
			&=	\Pr_\eta \left[ J^\eta(\bp,y) = i, J^\eta(\bp',y) = j\right]\\
			&= \int_\mathcal{H} \mathbbm{I} \left[ J^\eta(\mathbf{p}, y) = i \right] \mathbbm{I} \left[ J^\eta(\mathbf{p}', y) = j \right] d F(\eta),\\
		\end{split}
	\end{equation}
	which asymptotically coincide with the transition frequencies from alternative $i$ to alternative $j$ after an exogenous price change from $\mathbf{p}$ to $\mathbf{p}'$.\footnote{Note, however, that transition probabilities do not impose any temporal structure. In other words, $P_{i,j}(\mathbf{p}, \mathbf{p}', y) = P_{j,i}(\mathbf{p}', \mathbf{p}, y)$. Furthermore, as shown in Section \ref{sec:implementation}, the assumption that the exogenous income $y$ is common to both situations with prices $\bp$ and $\bp'$ imposes no constraint.} Naturally, if there is no price change, there are no transitions between different choices. In principle, these objects are nonparametrically estimable from panel data with at least two periods. In addition, Section \ref{sec:set} shows how transition probabilities can be set-identified when only cross-sectional data are available.

Implicit in our definition of the transition probabilities is the assumption that individuals' preferences are unaffected by the price change. The perfect correlation between the preference types before and after the price change implies that transition probabilities are not simply equal to the product of their marginals: i.e., $P_{i, j}(\mathbf{p},  \mathbf{p}', y) \neq P_{i}(\mathbf{p}, y) P_{j}(\mathbf{p}', y)$.

Finally, note that nonparametric estimates of the choice and transition probabilities can be evaluated at counterfactual prices and income, provided that these price-income pairs lie within the support of the data. This is particularly important for welfare estimation, which relies on knowledge of choice and transition probabilities at virtual prices (see Section~\ref{sec:distribution}).

\subsection{NOS welfare measures in a discrete choice setting}\label{ssec:metric_discrete}
In Section~\ref{sec:metrics_continuous}, the family of NOS welfare measures was introduced in a setting of continuously divisible goods. In this subsection, we will redefine them rigorously for settings where choice is determined by a DC-RUM that satisfies Assumptions~\ref{assumption1}--\ref{assumption3}. 

\paragraph{Nested opportunity sets in DC-RUMs.}
We define a family of nested opportunity sets, which is common for all preference types $\eta\in \mathcal{H}$, as follows. Let there be a closed set $\Lambda \subseteq \mathbb R$, and define for every $\lambda \in \Lambda$, 
an opportunity set $B_\lambda \subset \real \times \mathcal C$ by
\begin{equation}\label{eq:nested_budget1}
    B_\lambda \colonequals \left\{ (y',c) \mid  c\in\mathcal{C}, y' \in\real,y' \leq y_c^\lambda\right\},
\end{equation}
where $y_c^\lambda \in \mathbb R$ for all $c \in \C$ is satisfying the following assumptions:

\begin{enumerate}[(a)]
\item the function 
    \begin{equation}\label{eq:continuity}
      \Lambda \to \real :\lambda \mapsto y^{\lambda}_c  \text{ is continuous for all }c \in \C,
    \end{equation}
    \item \begin{equation}\label{eq:increasing}
\lambda < \lambda' \implies \begin{cases} \forall c\in \mathcal{C}:  y_{c}^{\lambda} \leq  y_{c}^{\lambda'}, \\ \exists c \in \mathcal{C}: y_{c}^{\lambda}  <  y_{c}^{\lambda'},\end{cases}
\end{equation}
     \item for all options $c'$, 
\begin{equation}\label{eq:infzero}
    \inf_{\lambda \in \Lambda} y^\lambda_{c'} = -\infty,
\end{equation}
    \item and for at least one option $c$, 
\begin{equation}\label{eq:supinfty}
    \sup_{\lambda \in \Lambda} y^\lambda_c = +\infty.
    \end{equation}
\end{enumerate}
Then $(B_\lambda)_{\lambda \in \Lambda}$ is called a family of nested opportunity sets and $\by^{\lambda} := (y^{\lambda}_{1},\ldots,y^{\lambda}_{c},\ldots,y^{\lambda}_{n})$ can be seen as its upper bound.\footnote{One can prove that every family of closed nested sets of which the option-wise suprema satisfy the four conditions above, are necessarily of the form \eqref{eq:nested_budget1}. Hence, the assumption that the $B_\lambda$ are of this form implies no loss of generality.} Note that the family is common to all individuals.

 Conditions \eqref{eq:continuity} and \eqref{eq:increasing} ensure that the family $(B_\lambda)_{\lambda \in \Lambda}$ is continuously increasing. Conditions \eqref{eq:infzero} and \eqref{eq:supinfty} together with Assumption \ref{assumption1} imply that for every bundle $\mathbf{b}$ and preference type $\eta$, there exists a member of the family of which all bundles are considered worse than $\mathbf{b}$ by $\eta$, and one which contains a bundle considered to be better than $\mathbf{b}$ by $\eta$. These properties will prove necessary to define the welfare measure. 

Figure~\ref{fig:NOSets_discrete} provides a graphical illustration. The choice set $\C$ consists of three options: $i$,$j$, and $k$. Two members of a family of nested opportunity sets $B_{\lambda \in \Lambda}$ are shown in gray. For example, all 
bundles in dark gray belong to $B_{\lambda_1}$. For illustrative convenience, we choose $y^\lambda_c < y^{\lambda'}_{c} $ for all $c\in \C$ whenever $\lambda < \lambda'$. Finally, the upper bounds of the $B_{\lambda}$'s, consisting of the points $y^{\lambda_g}_c$, 
($g=1,2$ and $c = i,j,k$), are denoted by the black dots.

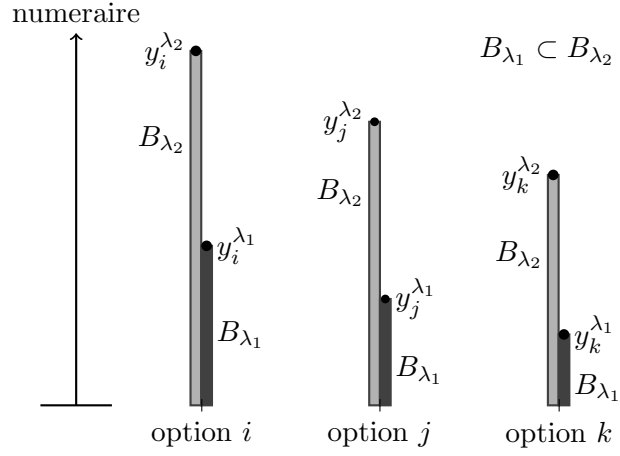
\begin{figure}[ht]
\center
\begin{tikzpicture}[scale=2.35,domain=0:1]
    \draw [->,thick] (-0.2,0) -- (-0.2,2.1) node [above] {numeraire} ;
   \draw [-,thick]  (-0.4,0) -- (0,0);

                       \node[ black ,right] at (2,2) {$B_{\lambda_1} \subset B_{\lambda_2}$} ;

       \filldraw [color=darkgray,thick,fill=darkgray!40,semitransparent]  (0.44,2) --  (0.50,2)--(0.50,0)--(0.44,0)--cycle;
       \filldraw [color=darkgray,thick,fill=darkgray!100,semitransparent]  (0.50,0.9) --  (0.56,0.9)--(0.56,0)--(0.50,0)--cycle;

       \filldraw [color=darkgray,thick,fill=darkgray!40,semitransparent]  (1.44,1.6) --  (1.50,1.6)--(1.50,0)--(1.44,0)--cycle;
       \filldraw [color=darkgray,thick,fill=darkgray!100,semitransparent]  (1.50,0.6) --  (1.56,0.6)--(1.56,0)--(1.50,0)--cycle;
            
       \filldraw [color=darkgray,thick,fill=darkgray!40,semitransparent]  (2.44,1.3) --  (2.50,1.3)--(2.50,0)--(2.44,0)--cycle;
       \filldraw [color=darkgray,thick,fill=darkgray!100,semitransparent]  (2.50,0.4) --  (2.56,0.4)--(2.56,0)--(2.50,0)--cycle;
    
               \node[ black ,right] at (0.53,0.4) {$B_{\lambda_1}$} ;
              \node[ black ,right] at (1.53,0.2) {$B_{\lambda_1}$} ;
               \node[ black ,right] at (2.53,0.1) {$B_{\lambda_1}$} ;

                      \node[ black,left] at (0.47,1.5) {$B_{\lambda_2}$} ;
                      \node[ black,left] at (1.47,1.2) {$B_{\lambda_2}$} ;
                      \node[ black,left] at (2.47,0.85 ) {$B_{\lambda_2}$} ;

 \fill[black] (0.47,2) circle (0.85pt) node [left,color = black] {$y^{\lambda_2}_{i}$} ;
 \fill[black]  (1.47,1.6) circle (0.7pt) node [left,color = black] {$y^{\lambda_2}_{j}$} ;
 \fill[black]  (2.47,1.3) circle (0.85pt) node [left,color = black] {$y^{\lambda_2}_{k}$}  ;
 
  \fill[black] (0.53,0.9) circle (0.85pt) node [right,color = black] {$y^{\lambda_1}_{i}$} ;
 \fill[black]  (1.53,0.6) circle (0.7pt) node [right,color = black] {$y^{\lambda_1}_{j}$} ;
 \fill[black]  (2.53,0.4) circle (0.85pt) node [right,color = black] {$y^{\lambda_1}_{k}$}  ;
   
    \draw [black](0.503,1pt) -- (0.503,-1pt) node[anchor=north] {option $i$};
    \draw [black](1.5,1pt) -- (1.5,-1pt) node[anchor=north] {option $j$};
    \draw [black](2.503,1pt) -- (2.503,-1pt) node[anchor=north] {option $k$};

\end{tikzpicture}
\caption{A graphical illustration of a family of nested opportunity sets in a discrete choice context. The  opportunity set $B_{\lambda_1}$ consists of the three dark gray spikes. $B_{\lambda_2}$ consists of the three light gray spikes.}
\label{fig:NOSets_discrete}
\end{figure}

It is often more convenient to characterize the opportunity sets in terms of virtual prices $\widetilde{p}_c(\lambda)\colonequals y - y^\lambda_c$ instead of the upper bounds $y^\lambda_c$. 
In particular, we have that
\begin{equation}\label{eq:nested_budget2}
    B_\lambda \colonequals \left\{ (y',c) \mid  c\in\mathcal{C}, y' \in\real,y' \leq y - \widetilde{p}_c(\lambda)\right\}.\footnote{It might be surprising that the --- individual independent --- opportunity sets $B_\lambda$ are linked with virtual prices $\tilde p_c(\lambda)$, which depend on individual incomes. However, in our discrete context, prices and incomes are only determined up to an additive constant as only $y-p_c$ enters the utility function, not $y$ nor $p_c$ separately. Hence, also $\tilde p_c(\lambda)$ must be defined such that the relevant concept $y- \tilde p_c(\lambda) $ is individual independent. Therefore $\tilde p_c(\lambda)$ itself must depend on individual income $y$.} 
\end{equation}
We denote the vector of virtual prices as follows: $\widetilde{\mathbf{p}}(\lambda) = \big(\widetilde p_1(\lambda),\ldots,\widetilde p_n(\lambda)\big)$. As $\by^\lambda$ is increasing in $\lambda$, in the sense of Equation~\eqref{eq:increasing}, $\widetilde{\bp}(\lambda)$ is decreasing in $\lambda$ in the same way. Moreover, $\lambda \mapsto \widetilde{\bp}(\lambda)$ is continuous by \eqref{eq:continuity},  $\sup_{\lambda \in \Lambda} \widetilde{p}_{c'}(\lambda) = +\infty$ for all  $c'$ by \eqref{eq:infzero} and $\inf_{\lambda \in \Lambda} \widetilde{p}_c(\lambda) = -\infty$ for at least one  $c$ by \eqref{eq:supinfty}. The fact that those virtual prices can become negative might seem odd at first. However, in a discrete choice context, one can always redefine prices and exogenous income by increasing both by an equal amount of the numeraire. As a result, negative prices can be converted into positive prices.

\paragraph{Welfare measures in DC-RUMs.}
In the continuous setting, the NOS welfare measure evaluated in a bundle was defined as the largest value of $\lambda$ ---or equivalently, the largest opportunity set $B_\lambda$--- such that this bundle was weakly preferred over all bundles in $B_\lambda$. The same idea can be applied to a setting where choice is discrete. More precisely, we define a NOS welfare measure as
\begin{equation}\label{eq:def_NOSWM_1}
    W^\eta(y-p_k,k) = \sup\Big\{\lambda \in \Lambda \mid U^{\eta}_{k}(y-p_k) \geq \max_{(y',c)\in B_\lambda}U^{\eta}_{c}(y')\Big\},
\end{equation}
that is, the largest value of $\lambda$ such that option $k$ is weakly preferred over all bundles in $B_\lambda$. Note that the dependence on the preference type $\eta$ implies that this welfare measure is a random variable.
According to Assumption~\ref{assumption1}, the utility function is strictly increasing in the numeraire, which allows us to restate this definition in terms of the upper bound of the opportunity sets. Formally, we have that
\begin{equation}\label{eq:def_NOSWM_1b}
\begin{split}
W^\eta(y-p_k,k) &= \sup \Big\{\lambda \in \Lambda \mid U^{\eta}_{k}(y-p_k) \geq \max_{(y',c)\in B_\lambda}U^{\eta}_{c}(y')\Big\}\\
&= \sup \Big\{\lambda \in \Lambda \mid U^{\eta}_{k}(y-p_k) \geq \max_{c } \max_{y' \leq  y^\lambda_c}U^{\eta}_{c}(y')\Big\}\\
&=\sup \Big\{\lambda \in \Lambda \mid U^{\eta}_{k}(y-p_k) \geq \max_{c}U^{\eta}_{c}(y^\lambda_c)\Big\}.
\end{split}
\end{equation}
This expression highlights that the value of the welfare measure only depends on the upper bound of the opportunity sets. Furthermore, by conditions~\eqref{eq:infzero} and~\eqref{eq:supinfty} and Assumption~(R1),  there exists (i) a  $\lambda_{\min  }\in \Lambda$ such that $U^{\eta}_{k}(y-p_k) \geq \max_c U^{\eta}_{c}(y^{\lambda_{\min}}_c)$, and (ii) a $\lambda_{\max  }\in \Lambda$ such that $U^{\eta}_{k}(y-p_k) < \max_c U^{\eta}_{c}(y^{\lambda_{\max}}_c)$. This implies that  the set $\Big\{\lambda \in \Lambda \mid U^{\eta}_{k}(y-p_k) \geq \max_{c}U^{\eta}_{c}(y^\lambda_c)\Big\}$ is not empty by (i), and bounded by (ii). Moreover, by continuity of the utility function and of the function $\lambda \mapsto \by^\lambda$, $\lambda \mapsto \max_c U^\eta_c(y^\lambda_c)$ is also continuous, which implies, together with the closedness of $\Lambda$, that (iii) $\Big\{\lambda \in \Lambda \mid U^{\eta}_{k}(y-p_k) \geq \max_{c}U^{\eta}_{c}(y^\lambda_c)\Big\}$ is closed. As this set is not empty, bounded and closed, one can conclude that the suprema in Equations \eqref{eq:def_NOSWM_1} and \eqref{eq:def_NOSWM_1b} are in fact attained and can be replaced by maxima.

Equivalently, when opportunity sets are characterized in terms of virtual prices, we can write that
\begin{equation}\label{eq:def_NOSWM_2}
	W^\eta(y-p_k,k) = \max \Big\{\lambda \in \Lambda \mid U^{\eta}_{k}(y-p_k) \geq \max_{c}U^{\eta}_{c}(y - \widetilde{p}_c(\lambda))\Big\}.
\end{equation}
For notational convenience, the characterization in terms of virtual prices $\widetilde{p}_c(\lambda)$ instead of the numeraire $y^\lambda_c$ will be used in the remainder of the paper.

A common feature of all NOS~measures is that statements on welfare levels, namely `type~$\eta$'s welfare obtained from a bundle $\left(y-p_k,k\right)$ is at least equal to $w$',  can be related to statements on optimal choice behavior, namely `$k$ is $\eta$'s best choice from a virtual opportunity set including all options~$c$ available at (welfare measure specific) virtual prices~$\widetilde{p}_c\left(w\right)$, and $k$ at the actual price, $p_k$.'\footnote{Notice that the welfare level of option $k$ at price $p_k$, $W^\eta(y-p_k,k)$, can never be greater or equal to~$w$ if $\widetilde{p}_k(w)<p_k$. The set of preference types defined in Equation~\eqref{eq:lem1} will be empty in this case.\label{note:maxwelfare}} 
This feature is key for all future results, and is made precise in Lemma~\ref{lem:characterisation_WB}.

\begin{lemma}\label{lem:characterisation_WB}
\begin{equation}\label{eq:lem1}
\big\{ \eta \mid w \leq W^\eta(y-p_k,k)\big\} = \big\{\eta \mid  U^\eta_k(y-p_k) \geq \max_{c} U^\eta_c(y-\widetilde{p}_c(w))\big\}
\end{equation} 
\end{lemma}
\begin{proof}
Take an arbitrary $\eta \in \mathcal{H}$ such that $w \leq W^\eta(y-p_k,k)$. Then there exists a $\lambda \geq w$ such that $U^\eta_k(y-p_k) \geq \max_{c} U^\eta_c(y-\widetilde{p}_c(\lambda))$. As $\lambda \geq w$, $\widetilde{p}_c(\lambda) \leq \widetilde{p}_c(w)$ for all $c\in \C$ and, hence,
$ \max_{c} U^\eta_c(y-\widetilde{p}_c(\lambda))\geq \max_{c} U^\eta_c(y-\widetilde{p}_c(w)),$ because $U^\eta_c$ is an increasing function by Assumption \ref{assumption1}. It follows that $U^\eta_k(y-p_k) \geq \underset{c}{\max}\, U^\eta_c(y-\widetilde{p}_c(w))$. The other inclusion follows immediately from the definition of $W^\eta(y-p_k,k)$.
\end{proof}
This equivalence is obtained without imposing any assumption on preferences besides~Assumptions~\ref{assumption1} and~\ref{assumption2}, and is, therefore, entirely nonparametric. As will be shown below, its main practical implication is that the entire distribution of objects based on NOS measures can be obtained by only evaluating choice and transition probabilities at virtual prices. This entails that these distributions can be identified from cross-sectional and panel data in a nonparametric way.

The illustration in Figure~\ref{fig:NOS_discrete} builds further on Figure~\ref{fig:NOSets_discrete}. For each option $i,j$, and $k$, the amount of the numeraire, $y-p_i$, $y-p_j$, and $y-p_k$ is shown on the vertical axis. The black circles indicate the indifference set of the point $\left(y-p_k,k\right)$. 
The figure shows how to calculate the NOS~welfare measure for option~$k$. The welfare measure is defined by the nested opportunity sets $B_{\lambda \in \Lambda}$ shown in gray in the figure. It is clear that $\lambda_2$ is the maximizer of Equation~\eqref{eq:def_NOSWM_1} because the black dot of $y^{\lambda_2}_{j}$ coincides with the black circle at position~$j$. This means that $U^\eta_{j}(y^{\lambda_2}_{j}) = U^\eta_k(y-p_k)$, and hence $W^\eta(y-p_k,k) = \lambda_2$.

\begin{figure}[ht]
\center
\begin{tikzpicture}[scale=2.35,domain=0:1]
    \draw [->,thick] (-0.2,0) -- (-0.2,2.27) node [above] {numeraire} ;
   \draw [-,thick]  (-0.4,0) -- (0,0);

       \filldraw [color=darkgray,thick,fill=darkgray!40,semitransparent]  (0.44,2) --  (0.50,2)--(0.50,0)--(0.44,0)--cycle;
       \filldraw [color=darkgray,thick,fill=darkgray!100,semitransparent]  (0.50,1) --  (0.56,1)--(0.56,0)--(0.50,0)--cycle;

       \filldraw [color=darkgray,thick,fill=darkgray!40,semitransparent]  (1.74,1.45) --  (1.80,1.45)--(1.80,0)--(1.74,0)--cycle;
       \filldraw [color=darkgray,thick,fill=darkgray!100,semitransparent]  (1.80,0.8) --  (1.86,0.8)--(1.86,0)--(1.80,0)--cycle;
            
        \filldraw [color=darkgray,thick,fill=darkgray!40,semitransparent]  (3.04,1.15) --  (3.10,1.15)--(3.10,0)--(3.04,0)--cycle;
       \filldraw [color=darkgray,thick,fill=darkgray!100,semitransparent]  (3.10,0.6) --  (3.16,0.6)--(3.16,0)--(3.10,0)--cycle;

                      \node[ black,right] at (0.53,0.55) {$B_{\lambda_1}$} ;
                      \node[ black,right] at (1.83,0.4) {$B_{\lambda_1}$} ;
                      \node[ black,right] at (3.13,0.25  ) {$B_{\lambda_1}$} ;
                           \node[ black,left] at (0.44,1.35) {$B_{\lambda_2}$} ;
                       \node[ black,left] at (1.74,1.05) {$B_{\lambda_2}$} ;
                       \node[ black,left] at (3.04,0.8) {$B_{\lambda_2}$} ;

 \fill[black] (0.47,2) circle (0.85pt) node [left,color = black] {$y^{\lambda_2}_{i}$} ;
 \fill[black]  (1.77,1.45) circle (0.7pt) node [left,color = black] {$y^{\lambda_2}_{j}$} ;
 \fill[black]  (3.07,1.15) circle (0.85pt) node [left,color = black] {$y^{\lambda_2}_{k}$}  ;
 
     \draw[black,very thick] (0.5,2.2) circle (1.2pt);
     \draw[black,very thick] (1.8,1.45) circle (1.2pt);
     \draw[black,very thick] (3.1,1.3) circle (1.2pt);
 
   \fill[black] (0.44,1.62) rectangle (0.56,1.68) node[above right , black] { $y-p_i$};
      \fill[black] (1.74,1.87) rectangle (1.86,1.93) node[above right , black] { $y-p_j$};
   \fill[black] (3.06,1.282) rectangle (3.14,1.318) node[above right , black] { $y-p_k$};

    \draw [black](0.503,1pt) -- (0.503,-1pt) node[anchor=north] {option $i$};
    \draw [black](1.8,1pt) -- (1.8,-1pt) node[anchor=north] {option $j$};
    \draw [black](3.103,1pt) -- (3.103,-1pt) node[anchor=north] {option $k$};


 
\end{tikzpicture}
\caption[A graphical illustration of a NOS welfare measure in a discrete choice context.]{A graphical illustration of a NOS welfare measure in a discrete choice context.\footnotemark ~The black circles denote the indifference set of option $k$ with residual numeraire $y-p_k$. The  opportunity set $B_{\lambda_1}$ consists of the three dark gray spikes. $B_{\lambda_2}$ consists of the three light gray spikes. Black bars indicate the actual choice set.}
\label{fig:NOS_discrete}
\end{figure}
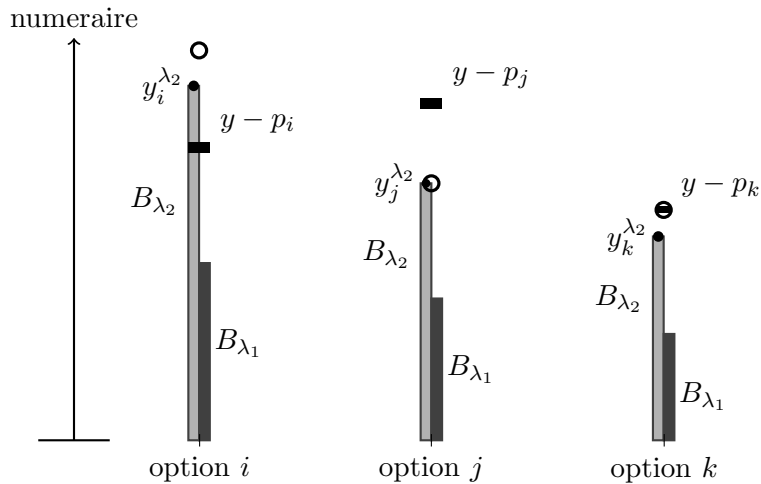
\footnotetext{Note that option~$k$ is not the chosen option as the bundle $(y-p_j,j)$ is situated above the indifference set of $(y-p_k,k)$ indicated by the black circles. However, in our approach, we also allow to calculate welfare in a non-optimal bundle.}

    Each member of the class of NOS~measures is characterized by a different family of nested budget sets. All NOS~measures agree on the welfare ranking of different bundles for a given  individual, as within-individual welfare rankings coincide with the individual's preference ordering over these bundles. Given that individual preferences are respected by all NOS~measures, they will also all agree whether a change in prices causes a welfare gain or loss for a given individual. They may disagree, however, on the making of interpersonal welfare comparisons.\footnote{The choice of the precise welfare metric is a normative choice. This remains true for the nonparametric approach advanced in this paper.}
    
    Moreover, the size of the gain or loss is not only dependent on the particular NOS~welfare measure, but also depends on the particular cardinalization of the size of the opportunity sets associated with that measure.\footnote{\cite{piacquadio17} proposes to choose the jointly least concave representation of this size.} The way the size of the opportunity sets of a given family (and thus of a given NOS~welfare measure) is measured, denoted by the parameter~$\lambda$, is only unique up to a strictly positive monotone transformation. The welfare ranking and interpersonal comparison of individuals by a given measure is not affected by the particular cardinalization of the sizes of the opportunity sets belonging to its associated family. If one is not willing to give any meaning to this cardinalization, only the sign of welfare differences is meaningful.\footnote{An exception is when two individuals are equally well-off before the price change. Who is the biggest gainer (or loser) in that case is again independent of the cardinalization, but may differ across NOS~measures.} Note that this is also true for measures of changes in welfare based on the MMU, like the CV and EV.

\paragraph{Examples of NOS~measures: the MMU~class.}
Let $\Lambda = \real$, fix a set of $n$~reference prices $\bp^{ref}$, one for each option, and let the upper bound of the opportunity sets be $y^\lambda_c = \lambda - p_c^{ref}$, or equivalently, $\widetilde{p}_c(\lambda) = y - \lambda + p_c^{ref}$. The crucial property is that the upper bounds increase by the same amount for every option: i.e., $y^{\lambda_1}_c - y^{\lambda_2}_c = \lambda_1 - \lambda_2$ for all $c \in \C$. The MMU evaluated in option $k$ (with price $p_k$) at reference prices $\bp^{ref}$ is then defined as
\begin{equation}\label{eq:MMUexpl}
\begin{array}{lcl}
W^\eta_{M(\bp^{ref})}\left(y-p_k,k\right) &=& {\max} \Big\{\lambda \in \real \mid U^{\eta}_{k}(y-p_k) \geq \underset{c}{\max}\,  U^{\eta}_{c}\left(y -(y -\lambda+p^{ref}_c)\right)\Big\}.
\end{array}
\end{equation}
This can also be defined implicitly as
\begin{equation}\label{eq:MMU}
U^\eta_k\left(y-p_k\right)=\max_{c}
U^\eta_c\left(W^\eta_{M(\bp^{ref})}\left(y-p_k,k\right)-p^{ref}_c\right).    
\end{equation}
Similar to the continuous case, this highlights the equivalence of the MMUs with the expenditure function representation of preferences, as each of them  evaluates the expenditure function at a given set of reference prices.

From Equation~\eqref{eq:MMUexpl}, it can be seen that $W^\eta_{M(\bp)}\left(y-p_k,k\right)=y$ if $k=J^\eta\left(\bp,y\right)$. When the reference prices coincide with the actual prices, the level of well-being according to the MMU of the optimal choice in the actual situation is equal to the actual amount of the numeraire (see also Corollary~\ref{cor:MMU} below).

\section{Distribution of welfare levels, welfare differences, and social welfare}
\label{sec:distribution}

As discussed before, the presence of unobserved preference heterogeneity entails that NOS welfare measures are random variables from the point of view of the econometrician. This randomness can be interpreted in two distinct ways. In the first interpretation, as the econometrician does not observe an individual's preference type, they can only derive the distribution of welfare for this particular individual and not its exact realization. That is, the distribution reflects uncertainty for the econometrician. In the second interpretation, an observed individual in the sample represents the class of individuals in the population that share the same observable characteristics. In this case, the distribution reflects inequality in welfare among the members of this class. Our theoretical results are valid for both interpretations.

For notational convenience, we will present all our expressions in terms of the complementary cumulative distribution function (CCDF) instead of the more common cumulative distribution function (i.e., $1-F(x)$ for a CDF $F$).\footnote{The CCDF is also known as the survival or reliability function.} The proofs are collected in the Appendix.

\subsection{Distribution of the NOS welfare measures}\label{ssec:distribution}
In this section the marginal distribution for the NOS measures is derived in terms of choice probabilities.\footnote{Similar to \cite{bhattacharyaNonparametricWelfareAnalysis2015}, it is important to stress that identification generally fails in settings where the prices of alternatives are multiples of one another, such as in ordered choice. In such settings, there is no relative price variation in the data that identifies the effect of a price change in some alternative(s) while keeping the prices of the other alternatives fixed.} We also provide distributional results joint with, and conditional on, the optimal observed choice.

Under Assumptions \ref{assumption1}--\ref{assumption3}, which were introduced in Section \ref{sec:formal_definitions}, we can prove the following theorem.
\begin{thmrep}\label{thm:distribution_arbitrary}
The joint distribution of the NOS welfare measure $W$, evaluated in an option $k$ with price $p_k$, and choosing $j$ at prices $\bp'$ and exogenous income $y$ can be expressed in terms of transition probabilities as follows:
\begin{equation}\label{eq:WB_distribution_arbitrary}
\Pr_\eta\big[w \leq W^\eta(y-p_k,k), j = J^\eta(\bp',y)\big]  =  P_{j,k}\Big(\mathbf{p}', \left(p_k,\widetilde{\mathbf{p}}_{-k}(w)\right), y\Big) \mathbb{I}\left[p_k \leq \widetilde p_k(w)\right],
\end{equation}
where $\left(p_k,\widetilde{\mathbf{p}}_{-k}(w)\right) = \left(\widetilde{p}_{1}(w),\ldots,\widetilde{p}_{k-1}(w),p_k,\widetilde{p}_{k+1}(w),\ldots,\widetilde{p}_{n}(w)\right)$.
\end{thmrep}

\begin{proofEnd}
Using Lemma~\ref{lem:characterisation_WB}, we have that
\begin{equation*}
\begin{split}
& \Pr_\eta\big[  w \leq W^\eta\left(y-p_k,k\right) ,  j= J^\eta(\mathbf p',y) \big]  \\
&=  \Pr_\eta\left[U^\eta_k(y-p_k) \geq \max_{c} U^\eta_c(y-\widetilde{p}_c(w)) , U^\eta_j(y-p'_j) \geq \max_{c'\neq j} U^\eta_{c'}(y-p'_{c'}) \right] \\
&= \Pr_\eta\left[U^\eta_k(y-p_k) \geq \max_{c\neq k} U^\eta_c(y-\widetilde{p}_c(w)) , U^\eta_j(y-p'_j) \geq \max_{c' \neq j} U^\eta_{c'}(y-p'_{c'}) \right]\mathbb{I}\left[p_k \leq  \widetilde p_k(w)\right] \\
&=  P_{j,k}\Big(\mathbf{p}', \left(p_k,\widetilde{\mathbf{p}}_{-k}(w)\right), y\Big) \mathbb{I}\left[p_k \leq \widetilde p_k(w)\right].
\end{split}
\end{equation*}
\end{proofEnd}

The theorem shows that the joint distribution can be expressed in terms of transition probabilities, evaluated at $p_k$, actual prices $\bp'$, and virtual prices~$\widetilde{\bp}$. The crucial insight here is that, by Lemma~\ref{lem:characterisation_WB}, the event $W^\eta(y-p_k,k) \geq w$ corresponds to the condition that $k$ is optimal under virtual prices. Consequently, the joint distribution in Equation~\eqref{eq:WB_distribution_arbitrary} can be written as
\begin{equation}
    \begin{split}
        &\Pr_\eta\left[  w \leq W^\eta\left(y-p_k,k\right) ,  j= J^\eta(\mathbf p',y) \right] \\
        &\quad = \int_\mathcal{H} \mathbb{I}[w \leq W^\eta\left(y-p_k,k\right)] \mathbb{I}\left[j= J^\eta(\mathbf p',y)\right] dF(\eta) \\
        &\quad = \int_\mathcal{H} \mathbb{I}\left[U^\eta_k(y-p_k) \geq \max_{c} U^\eta_c(y-\widetilde{p}_c(w))\right] \mathbb{I}\left[j= J^\eta(\mathbf p',y)\right] dF(\eta) \\
        &\quad = \mathbb{I}\left[p_k \leq \widetilde p_k(w)\right] \int_\mathcal{H} \mathbb{I}\left[U^\eta_k(y-p_k) \geq \max_{c \neq k} U^\eta_c(y-\widetilde{p}_c(w))\right] \mathbb{I}\left[j= J^\eta(\mathbf p',y)\right] dF(\eta) \\
        &\quad = \mathbb{I}\left[p_k \leq \widetilde p_k(w)\right] \int_\mathcal{H} \mathbb{I}\left[k= J^\eta(\left(p_k,\widetilde{\mathbf{p}}_{-k}(w)\right),y)\right] \mathbb{I}\left[j= J^\eta(\mathbf p',y)\right] dF(\eta) \\
        &\quad =  \mathbb{I}\left[p_k \leq \widetilde p_k(w)\right] P_{j,k}\Big(\mathbf{p}', \left(p_k,\widetilde{\mathbf{p}}_{-k}(w)\right), y\Big),
    \end{split}
\end{equation}
where the second equality follows from Lemma~\ref{lem:characterisation_WB}, the fourth equality from Assumption~\ref{assumption3}, and the fifth equality from the definition of the transition probabilities in Equation~\eqref{eq:transprob}.  The occurrence of the condition that the actual price of option $k$ should not be higher than the virtual price at welfare level~$w$ in the third equality can be understood from Footnote~\ref{note:maxwelfare}.

Theorem~\ref{thm:distribution_arbitrary} is formulated in the most general form; it considers a joint distribution, and not a marginal nor a conditional, and allows the price at which the welfare in alternative~$k$ is evaluated,~$p_k$, to be different from the actual prices~$\bp'$. For example, if one wants to evaluate welfare levels after a price change from~$\bp'$ to~$\bp$ when only information on choices before the price changed is available,~$\bp'$ and~$\bp$ will typically not coincide. However, in a setting where one wants to evaluate welfare at actual prices~$\bp'$, then $\bp$ equals $\bp'$. Usually, one wants to evaluate welfare in an optimal bundle; then~$k$ can be set equal to~$j$ in Equation~\eqref{eq:WB_distribution_arbitrary} (and $p_k$ equal to $p'_k$). In Corollary~\ref{cor:WB_derived_distribution} below, we will derive some related distributions which are more directly relevant for applied work.

The exact specification of $\widetilde{\bp}(w)$ depends, as explained in Section~\ref{ssec:metric_discrete}, on the specific choice of the welfare measure. Nonetheless, we can give some intuition on the role of $p_k$ and the overall course of the distribution of welfare. We know that the lower the price $p_k$, the higher is the residual numeraire $y-p_k$ in option $k$ and hence, the more the indifference set containing $(y-p_k,k)$ is shifted upwards in the numeraire dimension. A higher indifference set implies higher welfare, and therefore, the lower price $p_k$, the higher is the CCDF of welfare in option $k$ and the more the distribution of welfare is shifted to the right. 

Now, we examine the overall course of the CCDF in more detail by considering a typical plot of the CCDF for fixed prices $p_k$ and $\bp'$ in Figure~\ref{fig:CCDF}. When $w$ is negative and large in absolute value, the $\tilde p_{c}(w)$ are large (and positive). Hence $p_k \leq \tilde p_k(w)$ and the CCDF approaches $P_j(\bp')$ as expected. As $w$ increases, $\widetilde{\bp}_{-k}(w)$ decreases and the probability of choosing $k$ at prices $(p_k,\widetilde{\bp}_{-k}(w))$ decreases. Therefore, $\Pr_\eta\left[  w \leq W^\eta\left(y-p_k,k\right) ,  j= J^\eta(\mathbf p,y) \right] $ decreases until $w$ reaches its highest value, $w^*_k$,  where $p_k = \tilde p_k(w^*_k)$. There the CCDF drops to zero discontinuously, as the indicator function becomes zero. This means that $w^*_k$ is an upper bound for welfare and that the probability distribution has a mass point.

\begin{figure}[ht]
\center
\begin{subfigure}{0.5\textwidth}
\begin{tikzpicture}[scale=1.5,domain=0:1]
    \draw [->,thick] (0,0) -- (0,1.2) node [above] {CCDF} ;
    \draw (1pt, 0.78 cm) -- (-1pt, 0.78 cm) node[anchor=east] {$P_j(\bp')$};
   \draw [->,thick]  (0,0) -- (3,0)node [right] {welfare};
   \draw[thick, dashed, domain=0.15:0.3, smooth, variable=\x, darkgray] plot ({\x}, {(3-\x)/(3-\x+1)});
   \draw[thick, domain=0.3:2.5, smooth, variable=\x, darkgray] plot ({\x}, {(3-\x)/(3-\x+1)});
   \draw[thick,dotted, domain=0:1/3, smooth, variable=\y, darkgray] plot ({2.5}, {\y});
   \draw[thick, domain=2.5:3, smooth, variable=\x, darkgray] plot ({\x}, {0});
   
    \draw [black](2.5,1pt) -- (2.5,-1pt) node[anchor=north] {$w_k^*$};
\end{tikzpicture}
\caption{Unconditional CCDF}
\label{fig:CCDF}
\end{subfigure}%
\begin{subfigure}{0.5\textwidth}
    \center
\begin{tikzpicture}[scale=1.5,domain=0:1]
    \draw [->,thick] (0,0) -- (0,1.2) node [above] {CCDF} ;
    \draw (1pt, 1 ) -- (-1pt, 1) node[anchor=east] {1};
   \draw [->,thick]  (0,0) -- (3,0)node [right] {welfare};
   \draw[thick, dashed,domain=0.15:0.4, smooth, variable=\x, darkgray] plot ({\x}, {2*(3-1)/(3-1+1)-1/3});
      \draw[thick, domain=0.4:1, smooth, variable=\x, darkgray] plot ({\x}, {2*(3-1)/(3-1+1)-1/3});
   \draw[thick, domain=1:2.5, smooth, variable=\x, darkgray] plot ({\x}, {2*(3-\x)/(3-\x+1)-1/3});
   \draw[thick,dotted, domain=0:1/3, smooth, variable=\y, darkgray] plot ({2.5}, {\y});
   \draw[thick, domain=2.5:3, smooth, variable=\x, darkgray] plot ({\x}, {0});
   
    \draw [black](1,1pt) -- (1,-1pt) node[anchor=north] {$w^*$};
     \draw [black](2.5,1pt) -- (2.5,-1pt) node[anchor=north] {$w_k^*$};
\end{tikzpicture}
\caption{Conditional CCDF}
\label{fig:CCDF_conditonal}
\end{subfigure}
\caption{The course of the (un)conditional CCDF of welfare}
\end{figure}
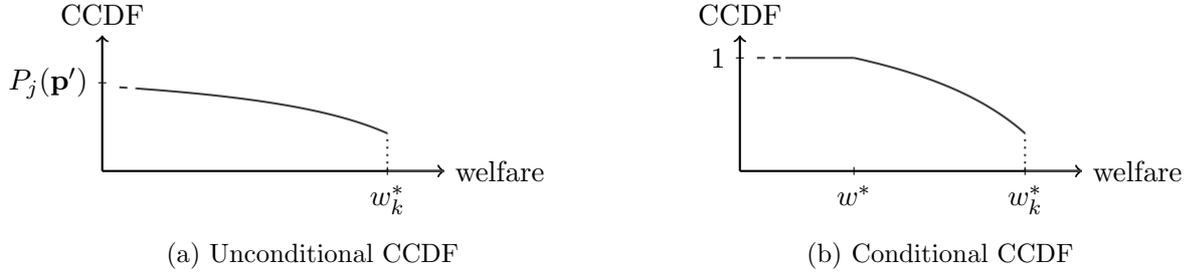

We can derive some associated distributions, such as the conditional and marginal CCDFs, which are more relevant in empirical applications. 
\begin{corollaryrep}\label{cor:WB_derived_distribution} 
\begin{equation}\label{eq:WB_conditional_distribution}
\Pr_\eta\Big[w \leq W^\eta\left(y-p_k,k\right)  \mid j = J^\eta(\mathbf p',y) \Big]  =  \frac{P_{j,k}\Big(\mathbf{p}', \left(p_k,\widetilde{\mathbf{p}}_{-k}(w)\right), y\Big) }{P_{j}\left(\mathbf{p}',y\right)} \mathbb{I}\left[p_k \leq \widetilde p_k(w)\right],
\end{equation}
\begin{equation}\label{eq:WB_conditional_distribution_chosen}
\Pr_\eta\Big[w \leq W^\eta\left(y-p_k, k\right)  \mid k = J^\eta(\mathbf p,y) \Big]  =  \frac{P_{k}\Big(\bmin\big( \mathbf{p}, \widetilde{\mathbf{p}}(w)\big), y\Big) }{P_{k}\left(\mathbf{p},y\right)} \mathbb{I}\left[p_k \leq \widetilde p_k(w)\right],
\end{equation}
where $\bmin\big( \mathbf{p}, \widetilde{\mathbf{p}}(w)\big) = \big( \min(p_1,  \widetilde{p}_1(w)),\ldots,\min(p_n, \widetilde{p}_n(w))\big),$ 
\begin{equation}\label{eq:WB_marginal_distribution}
\Pr_\eta\big[w \leq W^\eta\left(y-p_k,k\right)\big]  =  P_{k}\Big(\left(p_k,\widetilde{\mathbf{p}}_{-k}(w)\right),y\Big) \mathbb{I}\left[p_k \leq \widetilde p_k(w)\right],
\end{equation}
and 
\begin{equation}\label{eq:WB_marginal_distribution_chosen}
\Pr_\eta\Big[w \leq W^\eta\left(y-p_{J^\eta(\mathbf p,y)}, J^\eta(\mathbf p,y)\right)  \Big]  =  \sum_k P_{k}\Big(\bmin\big(\mathbf{p}, \widetilde{\mathbf{p}}(w)\big), y\Big) \mathbb{I}\left[p_k \leq \widetilde p_k(w)\right].
\end{equation}
\end{corollaryrep}

\begin{proofEnd}

Equations \eqref{eq:WB_conditional_distribution} and \eqref{eq:WB_marginal_distribution} follow directly from Theorem \ref{thm:distribution_arbitrary} using the definitions of conditional and marginal distributions respectively. Equation \eqref{eq:WB_marginal_distribution_chosen} follows analogously from Equation \eqref{eq:WB_conditional_distribution_chosen}. Therefore, we only prove Equation \eqref{eq:WB_conditional_distribution_chosen}. We have
\begin{equation*}
\begin{split}
&\Pr_\eta\Big[w \leq W^\eta(y-p_k,k)\mid  k= J^\eta(\mathbf p,y)\Big] \\
&=  \frac{\Pr_\eta\Big[w \leq W^\eta(y-p_k,k),  k= J^\eta(\mathbf p,y)\Big] }{ P_{k}\left(\mathbf{p},y\right)} \\
&= \frac{P_{k,k}\Big(\mathbf{p}, \left(p_k,\widetilde{\mathbf{p}}_{-k}(w)\right), y\Big)  \mathbb{I}\left[p_k \leq \widetilde p_k(w)\right]}{ P_{k}\left(\mathbf{p},y\right)}\\
&= \frac{\Pr_\eta\Big[U^\eta_k(y-p_k) \geq \max_{c\neq k}U^\eta_c(y-p_c),U^\eta_k(y-p_k) \geq \max_{c\neq k}U^\eta_c(y-\widetilde{p}_c(w) \Big]}{P_{k}\left(\mathbf{p},y\right)} \mathbb{I}\left[p_k \leq \widetilde p_k(w)\right]\\
&= \frac{\Pr_\eta\left[U^\eta_k(y-p_k) \geq \max_{c\neq k}U^\eta_c\big(y-\min(p_c,\widetilde{p}_c(w) )\big)\right]}{P_{k}\left(\mathbf{p},y\right)}\mathbb{I}\left[p_k \leq \widetilde p_k(w)\right]\\
&= \frac{ P_{k}\Big(\bmin\big(\mathbf{p}, \left(p_k,\widetilde{\mathbf{p}}_{-k}(w)\right)\big), y\Big)}{ P_{k}\left(\mathbf{p},y\right)}\mathbb{I}\left[p_k \leq \widetilde p_k(w)\right]\\
&=  \frac{P_{k}\Big(\bmin\big(\mathbf{p}, \widetilde{\mathbf{p}}(w)\big), y\Big) }{ P_{k}\left(\mathbf{p},y\right)}  \mathbb{I}\left[p_k \leq \widetilde p_k(w)\right].\\
\end{split}
\end{equation*}
\end{proofEnd}

We find again that the different derived distributions can be expressed in terms of choice and transition probabilities. Equations \eqref{eq:WB_conditional_distribution} and \eqref{eq:WB_conditional_distribution_chosen} can be used to assess the distribution of welfare when the researcher observes which bundle is optimal and wants to take this information into account. Equation \eqref{eq:WB_marginal_distribution} describes the marginal distribution of welfare evaluated in a specific bundle, not taking into account which bundle is optimal. Finally, Equation \eqref{eq:WB_marginal_distribution_chosen} specializes this result to a setting where welfare is evaluated in the optimal bundle.

A typical example of the distribution of welfare in bundle $k$ conditional on bundle $k$ being optimal is plotted in Figure \ref{fig:CCDF_conditonal}. As before, define $w^*_c$ to be the highest value of $w$ for an option $c$, such that $p_c = \tilde p_c(w)$, and let $w^*$ be $\min_{c} \{w^*_c\}$. Then we observe that for $w\leq w^*$, $p_c \leq \widetilde{p}_c(w)$ for all $c$, and hence, $\bmin(\bp,\widetilde{\bp}(w)) = \bp$. It follows that $\Pr_\eta\Big[w \leq W^\eta(y-p_k,k)\mid  k= J^\eta(\mathbf p,y)\Big] = 1$. Hence, $w^*$ is a lower bound for welfare in option $k$, conditionally on $k$ being optimal. For $w>w^*$, $\Pr_\eta\Big[w \leq W^\eta(y-p_k,k)\mid  k= J^\eta(\mathbf p,y)\Big]$ decreases continuously until $w$ reaches $w^*_k$ where $\Pr_\eta\Big[w \leq W^\eta(y-p_k,k)\mid  k= J^\eta(\mathbf p,y)\Big]$ drops to 0, as seen in Figure \ref{fig:CCDF_conditonal}. Hence, $w^*_k$ is an upper bound for welfare in option $k$, conditional on $k$ being optimal, and the distribution has a mass point at $w^*_k$. If $w^*_k = w^*$, the distribution is thus a step function and, hence, the welfare level in bundle~$k$, conditional on $k$ being optimal at prices~$\mathbf p$ and exogenous income $y$, is deterministic and equals~$w^*_k$. 

\paragraph{Examples of NOS measures: the MMU class.}
When applying Corollary \ref{cor:WB_derived_distribution} to the class of MMU measures, we obtain the following result.

\begin{corollaryrep}\label{cor:MMU}
When using reference prices $\bp^{ref}$, we have
\begin{equation}\label{eq:RW_joint_distribution}
\begin{split}
    &\Pr_\eta\Big[w \leq W^\eta_{M(\bp^{ref})}\left(y-p_k,k\right), j =  J^\eta(\bp',y) \Big]  = \\
    &\qquad P_{j,k}\big(\bp',(p_k,y-w + \bp^{ref}_{-k}),y\big)\mathbb{I}\left[p_k \leq y-w+p^{ref}_k \right].
\end{split}
\end{equation}
When $p_k = p_k'$, and the reference prices equal the actual prices $\bp'$ and $k$ is the optimal choice, this simplifies to
\begin{equation}\label{eq:RW_joint_distribution_chosen}
\Pr_\eta\Big[w \leq W^\eta_{M(\bp')}\left(y-p'_k,k\right)  , k = J^\eta(\mathbf p',y) \Big]  = P_{k}(\bp',y)\mathbb{I}\left[w \leq y\right]
\end{equation}
and, hence,
\begin{align}
    \Pr_\eta\Big[w \leq W^\eta_{M(\bp')}\left(y-p'_k,k\right)  \mid k = J^\eta(\mathbf p',y) \Big]  = \mathbb{I}\left[w \leq y\right],  \\
 \Pr_\eta\Big[w \leq W^\eta_{M(\bp')}\Big(y-p'_{J^\eta(\mathbf p',y)}, J^\eta(\mathbf p',y)\Big)   \Big]= \mathbb{I}\left[w \leq y\right].
\end{align}
\end{corollaryrep}

\begin{proofEnd}
    The first equation follows from plugging $\widetilde{\mathbf{p}}(w) = y-w + \mathbf{p}^{ref}$ into Equation~\eqref{eq:WB_distribution_arbitrary}. Moreover, using actual prices $\bp'$ as reference prices and taking $p_k = p'_k$, 
    $\mathbb{I}\left[p_k \leq y-w+p^{ref}_k \right]$ reduces to $\mathbb{I}\left[ w \leq y\right]$. Therefore,
    \begin{equation}
    \begin{split}
         \Pr_\eta\Big[w \leq W^\eta_{M(\bp')}\left(y-p'_k,k\right)  , k = J^\eta(\mathbf p',y) \Big] &=
     P_{k}\Big(\bmin\big( \mathbf{p}', y- w + \bp'\big), y\Big) \mathbb{I}\left[w \leq y\right]\\
     &=P_{k}(\bp', y) \mathbb{I}\left[w \leq y\right].
    \end{split}
 \end{equation}
     The last two equations then immediately follow from Bayes' theorem and summing over $k$.
\end{proofEnd}
Both the MMU in the optimal bundle and the MMU in bundle $k$, conditional on $k$ being optimal, are, therefore, deterministic and equal the initial exogenous income $y$ when reference equal actual prices.

\subsection{Joint distribution of welfare levels and welfare differences}
In this section, we derive the joint distribution of welfare levels and welfare differences. Joint knowledge on levels and differences of welfare enables investigation of the association between individuals' gains or losses from a price change and their position in terms of initial welfare. A price change is defined as an exogenous shift in prices from $\mathbf{p}$ to $\mathbf{p}'$. As discussed in Section \ref{ssec:DCM}, we will assume throughout that the unobserved preference type $\eta$ is unaltered by the price change.

\subsubsection{Welfare differences in terms of NOS measures} 
We first study the general case in which welfare differences are defined on the basis of changes in NOS welfare measures (evaluated in optimal choices). As an intermediate step, we derive the joint distribution of welfare before and after a price change in Proposition~\ref{prop:jointbeforeafter}.
\begin{propositionrep} 
\label{prop:jointbeforeafter}
The joint distribution of welfare in the optimal bundle $i$, before a price change, and welfare in the optimal bundle $j$, after the price change, is as follows:
\begin{equation}\label{eq:jointwelfare01}
	\begin{split}
&\Pr_\eta[w \leq W^\eta_0(y-p_i,i), z \leq W^\eta_1(y-p'_j, j) , i = J^\eta(\bp,y), j = J^\eta(\bp',y)] \\
&= P_{i,j}\Big( \bmin\big(\bp, \widetilde{\bp}(w)\big),\bmin\big(\bp', \widetilde{\bp}(z)\big),y\Big) \mathbb{I}\left[p_i \leq \widetilde{p}_i(w)\right]\mathbb{I}\left[p'_j \leq  \widetilde{p}_j(z)\right].
\end{split}
\end{equation}
\end{propositionrep}

\begin{proofEnd}
\begin{equation*}
	\begin{split}
	&\Pr_\eta[w \leq W^\eta_0(y-p_i,i), z \leq W^\eta_1(y-p'_j, j) , i = J^\eta(\bp,y), j = J^\eta(\bp',y)] \\
&= \Pr_\eta\Big[U^\eta_i(y-p_i) \geq \max_{c'} U^\eta_{c'}(y-\widetilde{p}_{c'}(w)), \quad U^\eta_i(y - p_i) \geq \max_{k\neq i}  U^\eta_k(y-p_k) ,\\
&	\qquad\qquad  U^\eta_j(y- p'_j) \geq \max_{l\neq j}  U^\eta_l(y-p'_l) , \quad U^\eta_j(y-p'_j) \geq \max_{c} U^\eta_{c}(y-\widetilde{p}_{c}(z)),\Big] \\
&= \Pr_\eta\Big[U^\eta_i(y - p_i) \geq \max_{k\neq i}  U^\eta_k(y-\min(p_k, \widetilde{p}_{k}(w))) , \\
&\qquad \qquad  U^\eta_j(y- p'_j) \geq \max_{l\neq j}  U^\eta_l(y-\min(p'_l, \widetilde{p}_{l}(z))) \Big] \mathbb{I}\left[p_i \leq \widetilde{p}_i(w)\right]\mathbb{I}\left[p'_j \leq  \widetilde{p}_j(z)\right]\\
&= P_{i,j}\Big(\big( p_i,\bmin\big(\bp_{-i}, \widetilde{\bp}_{-i}(w)\big) \big), \big(p_j',\bmin\big(\bp'_{-j}, \widetilde{\bp}_{-j}(z)\big),y\Big)  \mathbb{I}\left[p_i \leq \widetilde{p}_i(w)\right]\mathbb{I}\left[p'_j \leq  \widetilde{p}_j(z)\right]\\
&= P_{i,j}\Big( \bmin\big(\bp, \widetilde{\bp}(w)\big),\bmin\big(\bp', \widetilde{\bp}(z)\big),y\Big)  \mathbb{I}\left[p_i \leq \widetilde{p}_i(w)\right]\mathbb{I}\left[p'_j \leq  \widetilde{p}_j(z)\right].
\end{split}
\end{equation*}
\end{proofEnd}
Proposition~\ref{prop:jointbeforeafter} shows that this joint distribution can be written in terms of transition probabilities, evaluated at initial, final, and virtual prices. Using this proposition, the joint distribution of welfare levels and differences can be derived.
\begin{thmrep}\label{thm:joint_distribution_general_arbitrary}
Let the function $h$ be defined by
\begin{align}
     h_{i,j,\bp,\bp'}(w,x,s) &= P_{i,j}\Big( \bmin\big(\bp, \widetilde{\bp}(\max(w,x))\big),\bmin\big(\bp', \widetilde{\bp}(s)\big),y\Big)  \mathbb{I}\left[p'_j \leq  \widetilde{p}_j(s)\right]\\
     &= P_{i,j}\Big( \bmin\big(\bp, \widetilde{\bp}(w),\widetilde{\bp}(x)\big),\bmin\big(\bp', \widetilde{\bp}(s)\big),y\Big)  \mathbb{I}\left[p'_j \leq  \widetilde{p}_j(s)\right]. 
\end{align}

Then, the joint distribution of the stochastic welfare measure and the difference before and after the price change of this measure becomes,
\begin{equation}\label{joint_distribution_general}
	\begin{split}
&\Pr_\eta[w \leq W^\eta_0(y-p_i,i),  W^\eta_1(y-p'_j, j) -W^\eta_0(y-p_i,i)\leq z, i = J^\eta(\bp,y), j = J^\eta(\bp',y) ]= \\
& -\int_{-\infty}^{+\infty} \partial_{3} h_{i,j,\bp,\bp'}(w,x, x +z) \mathbb{I}\left[p_i \leq \min(\widetilde{p}_i(w),\widetilde{p}_i(x))\right]\,dx.
\end{split}
\end{equation}
\end{thmrep}

\begin{proofEnd}
Fix $i$ and $j$ and define $g(w,z) = \Pr_\eta[w \leq W^\eta_0(y-p_i,i), z \leq W^\eta_1(y-p'_j, j) , i = J^\eta(\bp,y), j = J^\eta(\bp',y)]$. Then we have

\begin{equation*}
\begin{split}
&\Pr_\eta[w \leq W^\eta_0(y-p_i,i),  W^\eta_1(y-p'_j, j) -W^\eta_0(y-p_i,i) \leq z, i = J^\eta(\bp,y), j = J^\eta(\bp',y) ] \\
&= - \int_{-\infty}^{+\infty} \partial_2 g\left( \max(w,x),x+z\right)\,dx\\
&= -\int_{-\infty}^{+\infty} \partial_{3} h_{i,j,\bp,\bp'}(w,x, x +z)\mathbb{I}\left[p_i \leq \widetilde{p}_i(\max(w,x))\right] \,dx\\
&= -\int_{-\infty}^{+\infty} \partial_{3} h_{i,j,\bp,\bp'}(w,x, x +z) \mathbb{I}\left[p_i \leq \min(\widetilde{p}_i(w),\widetilde{p}_i(x))\right]\,dx.
\end{split}
\end{equation*}
\end{proofEnd}
Unfortunately, it seems that this expression cannot be simplified. However, even though expression \eqref{joint_distribution_general} seems technically complicated, only the transition probabilities are used as input. This object is nonparametrically identified from panel data.

\subsubsection{Welfare differences in terms of the CV} 
We now specialize our results to the joint distribution of the MMU and the CV, which is a popular choice among applied welfare economists.\footnote{The results below in Theorems~\ref{thm:distribution_CV_initial_final} and~\ref{thm:joint_distribution_CV_arbitrary}, and in Corollaries~\ref{cor:CV} and~\ref{cor:WandCV}, can in fact be seen as applications of Theorem~\ref{thm:joint_distribution_general_arbitrary}. The derivation for the EV is similar and can be found in the Appendix in Section~\ref{app:add_results}.} The CV refers to the (possibly negative) amount of the numeraire an individual wants to give up after a price change to be equally well-off as before this change.

For an individual of type $\eta$, $CV^\eta$  is implicitly defined as 
\begin{equation}\label{eq:def_CV}
	\max_{c} \{U^\eta_c(y-p_c)\} = \max_{c} \{U^\eta_c(y-p_c' - CV^\eta)\},
\end{equation}
where  $\bp$ are initial prices and $\bp'$ final prices. In fact this definition of the CV is equivalent to $W^{\eta}_{M\left(\bp'\right)}(y-p'_{J^\eta(\bp',y)}, J^\eta(\bp',y)) - W^{\eta}_{M\left(\bp'\right)}(y-p_{J^\eta(\bp,y)}, J^\eta(\bp,y))$, i.e., the difference between the MMU with the final prices as reference price vector, in the optimal bundle after the price change, and the same MMU in the optimal bundle before the price change.\footnote{Indeed, defining $CV^\eta$ by $W^{\eta}_{M\left(\bp'\right)}(y-p'_{J^\eta(\bp',y)}, J^\eta(\bp',y)) - W^{\eta}_{M\left(\bp'\right)}(y-p_{J^\eta(\bp,y)}, J^\eta(\bp,y))$, we get $  W^{\eta}_{M\left(\bp'\right)}(y-p_{J^\eta(\bp,y)}, J^\eta(\bp,y)) = y - CV^\eta$ by Corollary \ref{cor:MMU}. Moreover, as $J^\eta(\bp,y)$ is the optimal bundle before the price change, we get
\begin{equation}
\begin{split}
      \max_{c} \{U^\eta_c(y-p_c)\} &=U^\eta_{J^\eta(\bp,y)}(y-p_{J^\eta(\bp,y)}) \\
      &=  \max_{c} \{U^\eta_c(W^{\eta}_{M\left(\bp'\right)}(y-p_{J^\eta(\bp,y)}, J^\eta(\bp,y))-p'_c)\} \\
      &= \max_{c} \{U^\eta_c(y-p_c' - CV^\eta)\}.
\end{split}
\end{equation}
}
Note that the CV for a composition of two or more price changes cannot be calculated from the CV for each price change separately. In our more general approach of measuring a change in welfare by the difference between two valuations of a welfare metric, this problem is inherently nonexistent. 

\paragraph{Distribution of the CV.}
In order to derive the distribution of the CV when the choice is equal to option $i$ under initial prices and option $j$ under final prices, we can follow a similar strategy as \citeauthor{bhattacharyaNonparametricWelfareAnalysis2015}~ (\citeyear{bhattacharyaNonparametricWelfareAnalysis2015}, \citeyear{bhattacharyaEmpiricalWelfareAnalysis2018}) and \cite{depalmaTransitionChoiceProbabilities2011}. Analogously to Lemma~\ref{lem:characterisation_WB}, the condition $CV^\eta \leq z, i = J^\eta(\bp,y), j = J^\eta(\bp',y)$ can be translated in $i$ being the optimal bundle when faced with a counterfactual price vector.
\begin{lemmarep}\label{lem:characterisation_CV}
We have
\begin{equation}
	\begin{split}
		&\Big\{\eta  \mid CV^\eta \leq z, i = J^\eta(\bp,y), j = J^\eta(\bp',y) \Big\}\\
&= 	\Big\{\eta  \mid U^\eta_i(y-p_i) \geq \max_{c } \{U^\eta_c(y-p_c' - z)\}, i = J^\eta(\bp,y), j = J^\eta(\bp',y) \Big\}.
	 	\end{split}
\end{equation} 
\end{lemmarep}

\begin{proofEnd}
\begin{equation*}
	\begin{split}
		&\Big\{\eta  \mid CV^\eta \leq z, i = J^\eta(\bp,y), j = J^\eta(\bp',y) \Big\} \\
		&=\Big\{\eta  \mid \max_{c } \{U^\eta_c(y-p_c' - CV^\eta)\} \geq \max_{c } \{U^\eta_c(y-p_c' - z)\}, i = J^\eta(\bp,y), j = J^\eta(\bp',y) \Big\}\\
		&=\Big\{\eta  \mid 	\max_{c} \{U^\eta_c(y-p_c)\}\geq \max_{c} \{U^\eta_c(y-p_c' - z)\}, i = J^\eta(\bp,y), j = J^\eta(\bp',y) \Big\}\\
&= 	\Big\{\eta  \mid U^\eta_i(y-p_i) \geq \max_{c} \{U^\eta_c(y-p_c' - z)\}, i = J^\eta(\bp,y), j = J^\eta(\bp',y) \Big\},
	 	\end{split}
\end{equation*} 
where the second equality follows from \eqref{eq:def_CV} and the last from $i=J^\eta\left(\bp,y\right)$. 
\end{proofEnd}

With Lemma~\ref{lem:characterisation_CV}, we can state the following theorem.

\begin{thmrep}\label{thm:distribution_CV_initial_final}
The joint distribution of the CV and the optimal choices before and after the price change is as follows:
\begin{equation}\label{eq:CV_joint_distribution}
\Pr_\eta[CV^\eta \leq z, i =J^\eta(\bp,y), j = J^\eta(\bp',y)] = P_{i,j}( \bmin(\bp,\bp' + z ), \bp',y) \mathbb{I}\left[p_i \leq p_i'+ z\right].
\end{equation}
\end{thmrep}

\begin{proofEnd}
We have
\begin{equation*}
	\begin{split}
	&\Pr_\eta[CV^\eta \leq z, i =J^\eta(\bp,y), j = J^\eta(\bp',y)]\\
	 &= \Pr_\eta\Big[U^\eta_i(y - p_i) \geq \max_{k\neq i}  U^\eta_k(y-p_k) ,\quad  U^\eta_j(y- p'_j) \geq \max_{l\neq j}  U^\eta_l(y-p'_l) , \\
&	\qquad\qquad\qquad U^\eta_i(y-p_i) \geq \max_{c} U^\eta_c(y-p_c' - z)\Big] \\
&= \Pr_\eta\Big[U^\eta_i(y - p_i) \geq \max_{k\neq i}  U^\eta_k(y-\min(p_k,p_k' + z )) ,\quad U^\eta_j(y- p'_j) \geq \max_{l\neq j}  U^\eta_l(y-p'_l) \Big]  \\
&\qquad \qquad\qquad\mathbb{I}\left[p_i \leq p_i' + z\right]\\
&= P_{i,j}( \bmin(\bp,\bp' + z ), \bp',y) \mathbb{I}\left[p_i \leq p_i'+ z\right].
\end{split}\end{equation*}
\end{proofEnd}

We observe that $\underset{\eta}{\Pr}[CV^\eta \leq z, i =J^\eta(\bp,y), j = J^\eta(\bp',y)]$ is bounded from below by $p_i-p_{i}'$. This is as expected; if the initial optimal bundle was $i$ and the price of $p_i$ drops to $p_i'$, the numeraire must drop with at least this amount to be equally well-off as in the initial situation. This means that the minimal compensation, in terms of the joint distribution, is $p_i-p_i'$. Moreover, for $z \geq \max_{k}\{p_k - p'_k\}$, $\underset{\eta}{\Pr}[CV^\eta \leq z\mid i =J^\eta(\bp,y), j = J^\eta(\bp',y)] =1$. This means that the maximal compensation, in terms of the conditional distribution, cannot be higher than the maximal price difference, which is also as expected.

The next corollary follows immediately and may again be more useful to the applied researcher.

\begin{corollary} \label{cor:CV}
\begin{equation}\label{eq:CV_cond_pre_post}
\Pr_\eta[CV^\eta \leq z \mid i =J^\eta(\bp,y), j =J^\eta(\bp',y)]  =  \frac{P_{i,j}( \bmin(\bp,\bp' + z ), \bp',y)}{P_{i,j}( \bp,\bp',y)} \mathbb{I}\left[p_i \leq p_i'+ z\right] ,
\end{equation}
\begin{equation}\label{eq:CV_cond_pre}
\Pr_\eta[CV^\eta \leq z \mid i =J^\eta(\bp,y)]  =  \frac{P_{i}( \bmin(\bp,\bp' + z ),y)}{P_{i}( \bp,y)} \mathbb{I}\left[p_i \leq p_i'+ z\right],
\end{equation}
\begin{equation}\label{eq:CV_cond_post}
\Pr_\eta[CV^\eta \leq z \mid  j = J^\eta(\bp',y)]  = \sum_{i}\frac{P_{i,j}( \bmin(\bp,\bp' + z ), \bp',y)}{P_{j}( \bp',y)} \mathbb{I}\left[p_i \leq p_i' + z \right],
\end{equation}
and
\begin{equation}\label{eq:CV_marginal}
\Pr_\eta[CV^\eta \leq z]  = \sum_{i} P_{i}( \bmin(\bp,\bp' + z ),y) \mathbb{I}\left[p_i \leq p_i'+ z\right].\footnote{Note that Equation~\eqref{eq:CV_marginal} is the main result of \cite{bhattacharyaNonparametricWelfareAnalysis2015}, which is a special case in our setup.}
\end{equation}
\end{corollary}
Equation~\eqref{eq:CV_marginal} gives an expression for the marginal distribution of the CV. Equations~\eqref{eq:CV_cond_pre_post},~\eqref{eq:CV_cond_pre}, and~\eqref{eq:CV_cond_post}, which present conditional distributions, can be used to calculate the distribution of the CV when the optimal bundle(s) (i)~ before and after price change are known; (ii)~only before the price change is known; and (iii)~only after the price change is known.

\paragraph{Joint distribution of the MMU and the CV.}
We now apply Theorem \ref{thm:joint_distribution_general_arbitrary} to the case where one chooses the MMU with final prices as the reference price vector, as a welfare measure. The difference in welfare before and after the price change is then equal to  the CV.
\begin{thmrep}\label{thm:joint_distribution_CV_arbitrary}
The joint distribution of the MMU with reference prices $\bp'$ and the CV is as follows:
\begin{equation}\label{eq:joint_distribution_CV_arbitrary}
\begin{split}
\Pr_\eta[w \leq W^\eta_{M(\bp')}(y-p_i,i), CV^\eta \leq z , i = J^\eta(\bp,y), j = J^\eta(\bp',y)] \\
= P_{i,j}\Big( \bmin\big(\bp,\bp' + \min(z, y-w) \big),\bp',y\Big) \mathbb{I}\left[p_i \leq p_i'+ \min(z, y-w)\right].
\end{split}
\end{equation}
\end{thmrep}

\begin{proofEnd}
\textbf{A direct proof of Theorem \ref{thm:joint_distribution_CV_arbitrary} }\\

We have
\begin{equation*}
	\begin{split}
	&\Pr_\eta[w \leq W^\eta_{M(\bp')}(y-p_i,i), CV^\eta \leq z, i = J^\eta(\bp,y), j = J^\eta(\bp',y)] \\
&= \Pr_\eta\Big[U^\eta_i(y-p_i) \geq \max_{c'} U^\eta_{c'}(y- (y -w + p'_{c'}) ), \quad U^\eta_i(y - p_i) \geq \max_{k\neq i}  U^\eta_k(y-p_k) ,\\
&	\qquad\qquad  U^\eta_j(y- p'_j) \geq \max_{l\neq j}  U^\eta_l(y-p'_l) , \quad U^\eta_i(y-p_i) \geq \max_{c } U^\eta_c(y-p_c' - z)\Big] \\
&= \Pr_\eta\Big[U^\eta_i(y - p_i) \geq \max_{k\neq i}  U^\eta_k(y-\min(p_k, p'_k + y- w, p_k' + z )) , \\
&\qquad \qquad  U^\eta_j(y- p'_j) \geq \max_{l\neq j}  U^\eta_l(y-p'_l) \Big]\mathbb{I}\left[p_i \leq p_i' + z\right]\mathbb{I}\left[p_i \leq  p_i' + y - w\right]\\
&= P_{i,j}\Big( \big(p_i,\bmin\big(\bp_{-i},\bp_{-i}' + \min(z ,y-w)\big)\big),\bp',y\Big) \mathbb{I}\left[p_i \leq p_i'+  \min(z ,y-w)\right]\\
&= P_{i,j}\Big( \bmin\big(\bp,\bp' + \min(z ,y-w)\big),\bp',y\Big) \mathbb{I}\left[p_i \leq p_i'+  \min(z ,y-w)\right].
\end{split}
\end{equation*}

\textbf{Theorem \ref{thm:joint_distribution_CV_arbitrary} as implied by Theorem \ref{thm:joint_distribution_general_arbitrary}}\\

When choosing the MMU with the final prices as reference prices, Theorem~\ref{thm:joint_distribution_general_arbitrary} implies:
\begin{equation}\label{eq:thm2_applied_to_MMU}
\begin{split}
     &\Pr_\eta[w \leq W^\eta_{M(\bp')}(y-p_i,i), CV^\eta \leq z , i = J^\eta(\bp,y), j = J^\eta(\bp',y)] \\
     &=-\int_{-\infty}^{+\infty} \partial_{3} h_{i,j,\bp,\bp'}(w,x, x +z) \mathbb{I}\left[p_i \leq \min(p'_i + y - w, p'_i + y - x)\right]\,dx.
\end{split}
\end{equation}
where the function $h$ is defined by 
\begin{equation}\label{eq:h_for_CV}
     h_{i,j,\bp,\bp'}(w,x,s) = P_{i,j}\Big( \bmin\big(\bp, \bp' + y -\max(w,x)\big),\bmin\big(\bp', \bp' + y -s \big),y\Big)  \mathbb{I}\left[p'_j \leq p'_j + y -s\right].
\end{equation}
Rewriting, \eqref{eq:h_for_CV}, we obtain
    \begin{align*}
     h_{i,j,\bp,\bp'}(w,x,s) &= P_{i,j}\Big( \bmin\big(\bp, \bp' + y -\max(w,x)\big),\bmin\big(\bp', \bp' + y -s \big),y\Big)  \mathbb{I}\left[p'_j \leq p'_j + y -s\right]\\
     &=P_{i,j}\Big( \bmin\big(\bp, \bp' + y -\max(w,x)\big),\bp',y\Big)  \mathbb{I}\left[s \leq y \right],
\end{align*}
and hence 
\begin{equation*}
      \partial_{3} h_{i,j,\bp,\bp'}(w,x, x +z) = - P_{i,j}\Big( \bmin\big(\bp, \bp' + y -\max(w,x)\big),\bp',y\Big)  \delta(x+z - y),
\end{equation*}
where $\delta$ is a Dirac delta function. Plugging this in in \eqref{eq:thm2_applied_to_MMU}, we obtain
\begin{align*}
     &\Pr_\eta[w \leq W^\eta_{M(\bp')}(y-p_i,i), CV^\eta \leq z , i = J^\eta(\bp,y), j = J^\eta(\bp',y)] \\
     &=-\int_{-\infty}^{+\infty} \partial_{3} h_{i,j,\bp,\bp'}(w,x, x +z) \mathbb{I}\left[p_i \leq \min(p'_i + y - w, p'_i + y - x)\right]\,dx\\
     &= P_{i,j}\Big( \bmin\big(\bp, \bp' + y -\max(w,y-z)\big),\bp',y\Big)\mathbb{I}\left[p_i \leq \min(p'_i + y - w, p'_i + y - (y-z))\right]\\
     &= P_{i,j}\Big( \bmin\big(\bp, \bp' + \min(y-w,z)\big),\bp',y\Big)\mathbb{I}\left[p_i \leq p'_i + \min(y - w,z)\right]
\end{align*}
as in Theorem \ref{thm:joint_distribution_CV_arbitrary}.
\end{proofEnd}

Again, Corollary \ref{cor:WandCV} follows immediately. 
\begin{corollary}\label{cor:WandCV}
\begin{multline}
\Pr_\eta\left[w \leq W^\eta_{M(\bp')}(y-p_i,i),CV^\eta \leq z \mid i =J^\eta(\bp,y) , j = J^\eta(\bp',y)\right] \\
= \frac{P_{i,j}\Big( \bmin\big(\bp,\bp' + \min(z ,y-w)\big),\bp',y\Big)}{P_{i,j}( \bp,\bp',y)} \mathbb{I}\left[p_i \leq p_i'+ \min(z ,y-w)\right],
\end{multline}
\begin{multline}
\Pr_\eta[w \leq W^\eta_{M(\bp')}(y-p_i,i),CV^\eta \leq z \mid i =J^\eta(\bp,y)]  \\
=  \frac{P_{i}\Big( \bmin\big(\bp,\bp' + \min(z ,y-w)\big),y\Big)}{P_{i}( \bp,y)} \mathbb{I}\left[p_i \leq p_i'+ \min(z ,y-w)\right],
\end{multline}
\begin{multline}
\Pr_\eta[w \leq W^\eta_{M(\bp')} \Big(y-p_{J^\eta(\bp,y)}, J^\eta(\bp,y)\Big),CV^\eta \leq z \mid j = J^\eta(\bp',y)] \\
= \sum_{i}\frac{P_{i,j}\Big( \bmin\big(\bp,\bp' + \min(z ,y-w)\big),\bp',y\Big) }{P_{j}( \bp',y)}\mathbb{I}\left[p_i \leq p_i'+ \min(z ,y-w)\right],
\end{multline}
and, \begin{multline}\label{eq:WandCV}
\Pr_\eta[w \leq W^\eta_{M(\bp')} \Big(y-p_{J^\eta(\bp,y)}, J^\eta(\bp,y)\Big) ,CV^\eta \leq z]  \\
= \sum_i
P_{i}\Big( \bmin\big(\bp,\bp' + \min(z ,y-w)\big),y\Big) \mathbb{I}\left[p_i \leq p_i'+ \min(z ,y-w)\right].
\end{multline}
\end{corollary}

The joint cumulative distribution can again be written as (a sum of) choice or transition probabilities. Each choice and transition probability is calculated using up to three price vectors: the initial price vector $\bp$, the final price vector $\bp'$, and a translation of the $\bp'$ vector for the combined MMU and CV part.

\subsection{Social welfare}

A classical additively separable Bergson-Samuelson social welfare function (SWF) takes the form
\begin{equation}\label{eq:SWF_BS}
SWF = \int h(u) \,dG_U(u),
\end{equation}
where $u$ is the value of a utility function representing the well-being of an individual in a particular state of the world,  $h$ is a strictly increasing, concave function expressing the inequality aversion, and $G_U$ is the CDF of the well-being distribution in the population in a given state of the world.\footnote{When $h$ is strictly concave, the Bergson-Samuelson SWF is also referred to as a \emph{prioritarian} SWF \citep{Adler_Norheim_2022}.} For example, in the utilitarian case, we have that $h(u) = u$. 

The NOS welfare measures are well suited as a representation of preferences as they are known to satisfy a set of attractive principles of interpersonal comparability (see \citeauthor{Fleurbaey2017}, \citeyear{Fleurbaey2017}; \citeyear{FlM2018}). We can, therefore, use these measures directly as building blocks in the SWF in Equation~\eqref{eq:SWF_BS}. More specifically, the  equivalent to the Bergson-Samuelson SWF in our framework reads as
\begin{equation}\label{eq:SWF_general2}
SWF = \int  \int h(w) \,d F_W(w\mid \bp,y) \,dG(\bp,y),
\end{equation}
where $G$ is the CDF of the joint distribution of prices and exogenous income in the population, which can be observed from the data, and $F_W(w\mid \bp,y)$ is the conditional CDF of the NOS measure $W$, and equals $\underset{\eta}{\mathrm{Pr}}\left[W^\eta\left(y-p_{J^\eta\left(\bp,y\right)}, J^\eta\left(\bp,y\right)\right)\leq w\right]$.\footnote{As social welfare is a population level concept, we rely on the second interpretation of the randomness in the welfare measure (see the discussion at the beginning of Section~\ref{sec:distribution}).}

Proposition~\ref{prop:SWF} illustrates how the results on the distribution of individual welfare levels in Corollary~\ref{cor:WB_derived_distribution} lead to the calculation of social welfare as defined in Equation~\eqref{eq:SWF_general2}, using only choice probabilities.

\begin{propositionrep}\label{prop:SWF}
The conditional CDF of individual welfare in the optimal bundle can be calculated using choice probabilities:
\begin{equation}\label{eq:cond_CDF_welfare}
     F_W(w\mid \bp,y) = 1 - \sum_{k}  P_{k}\Big(\bmin\big( \mathbf{p}, \widetilde{\mathbf{p}}(w)\big), y\Big)  \mathbb{I}\left[p_k \leq \widetilde p_k(w)\right].
\end{equation}
\end{propositionrep}
\begin{proofEnd}
\begin{equation}\label{eq:cond_CDF_welfare_proof}
\begin{split}
     F_W(w\mid \bp,y) &= \underset{\eta}{\mathrm{Pr}}\left[W^\eta\left(y-p_{J^\eta\left(\bp,y\right)}, J^\eta\left(\bp,y\right)\right)\leq w\right]\\
     &=1-\underset{\eta}{\mathrm{Pr}}\left[w\leq W^\eta\left(y-p_{J^\eta\left(\bp,y\right)}, J^\eta\left(\bp,y\right)\right)\right]\\
      &= 1 - \sum_{k}  P_{k}\Big(\bmin\big( \mathbf{p}, \widetilde{\mathbf{p}}(w)\big), y\Big)  \mathbb{I}\left[p_k \leq \widetilde p_k(w)\right],
\end{split}
\end{equation}
where the last equality follows from  Equation~\eqref{eq:WB_marginal_distribution_chosen} in Corollary~\ref{cor:WB_derived_distribution}.
\end{proofEnd}

Hence, social welfare can be computed from these probabilities. The joint distribution of prices and exogenous income $G$ can be estimated separately using standard nonparametric tools.

Moreover, this expression can be used to identify if a price change, for example, due to a policy reform, has a desirable effect on social welfare. Indeed, the difference in social welfare can be calculated as follows:
\begin{equation}\label{eq:diff_SWF}
\begin{split}
SWF' - SWF &= \int  \int h(w) \,d F_W(w\mid \bp',y) \,dG'(\bp',y) -\int  \int h(w) \,d F_W(w\mid \bp,y) \,dG(\bp,y) \\
&=\int  \int h(w) \,d F_W(w\mid \bp + \Delta\bp,y) \,dG'(\bp + \Delta\bp,y) \\
&\qquad \qquad -\int  \int h(w) \,d F_W(w\mid \bp,y) \,dG(\bp,y) \\
&=\int  \int h(w) \,d\Big( F_W(w\mid \bp + \Delta\bp,y) - F_W(w\mid \bp,y)\Big) \,dG(\bp,y). \\
\end{split}
\end{equation}
where $G$ ($G'$) is the joint distribution of initial (final) prices and exogenous income, and $\Delta \bp = \bp' - \bp$. With  Equations \eqref{eq:diff_SWF} and \eqref{eq:cond_CDF_welfare}, one can assess the desirability of a potential price change without parametric assumptions and only using choice probabilities and the initial distribution of prices and exogenous income.

Interestingly, in the spirit of~\cite{Roberts1980priceindependentwf}, we can derive conditions under which the expression for the SWF can be formulated in terms of incomes alone. In particular, when prices are equal for everyone and one uses the MMU with reference prices equal to those common prices, as individual welfare measure, one obtains a price independent SWF in terms of income. 

\begin{corollaryrep}
When prices are equal for everyone and when one uses the MMU with reference prices equal to those common prices  as the welfare measure, the SWF can be written solely in terms of income. 
\end{corollaryrep}

\begin{proofEnd}
From Proposition~\ref{prop:SWF} and the definition of the virtual prices in case of an MMU with actual prices~$\bp$ as reference prices ($\widetilde{\bp}(w)=y-w+\bp$), it follows that 
\begin{equation*}
\begin{split}
F_W(w\mid \bp,y) &= 1 -
\sum_{k}  P_{k}\Big(\bmin\big( \mathbf{p}, \widetilde{\mathbf{p}}(w)\big), y\Big)  \mathbb{I}\left[p_k \leq \widetilde p_k(w)\right]\\
&=1-\sum_{k}  P_{k}\Big(\bmin\big( \mathbf{p}, y-w+\bp\big), y\Big)  \mathbb{I}\left[p_k \leq y-w+p_k)\right]
\\
&=1 -\sum_{k}  P_{k}\Big(\bp,y\Big)  \mathbb{I}\left[w\leq y\right]\\
&=\mathbb{I}\left[y\leq w\right].
\end{split}\end{equation*}

Hence,
\begin{equation*}
\begin{split}
SWF &= \int  \int h(w) \,d F_W(w\mid \bp,y) \,dG(\bp,y)\\
&= \int  \int h(w) d\mathbb{I}\left[y\leq  w \right] \,dG(\bp,y)\\
&= \int h(y)\,dG(\bp,y).
\end{split}
\end{equation*}
Notice that $\bp$ in the argument of $G$ is redundant, as prices are assumed to be identical for all persons in this case. This completes the proof.
\end{proofEnd}

\section{Discussion on implementation}\label{sec:implement}
We now outline the practical implementation of our results in the common scenario where the analyst has access to cross-sectional data. The proposed workflow consists of the following steps.

\begin{itemize}[leftmargin=20mm]
    \item[\textbf{Step 1:}] \textbf{Estimate the choice probabilities} \\
    Begin by estimating the choice probabilities from observed data on choices, prices, and income. This estimation can be conducted using parametric, semi-parametric, or non-parametric methods, as detailed in Section~\ref{sec:estchoiceprob}.
    \item[\textbf{Step 2:}]    
    \textbf{Bound the transition probabilities from the choice probabilities} \\
    Next, derive bounds for the transition probabilities based on the choice probabilities. This step employs Boole-Fr\'echet inequalities and revealed preference constraints, as discussed in Section~\ref{sec:set}.
    \item[\textbf{Step 3:}] \textbf{Estimate the distributional welfare effects} \\
    Finally, compute the distributional welfare effects by substituting the estimates from Steps 1 and 2 into the results provided in Section~\ref{sec:distribution}. These welfare estimates are consistent by the plug-in principle. To enable statistical inference, repeat Steps 1–3 on random bootstrap samples drawn with replacement from the original dataset.
\end{itemize}

\subsection{Estimating the choice probabilities}\label{sec:estchoiceprob}
Given the exogeneity of budget sets presupposed in Assumption \ref{assumption2}, the choice probabilities can be readily estimated using nonparametric regression, as they are essentially conditional expectation functions. Standard tools, such as kernel and series based regression, are available in most modern statistical software. One particular attractive feature of the Nadaraya-Watson kernel estimator is that the estimated choice probabilities add up to one for all price vectors when the same bandwidth is selected for every choice probability function. With samples of modest size, it might be useful to impose additional structure to mitigate the curse of dimensionality. In particular, in a setting with high-dimensional regressors, which arises when there are many alternatives or many individual-specific characteristics, a (semi)parametric estimator can be used to increase efficiency at the expense of functional form misspecification. A popular parametric specification is the (nested) multinomial logit model (e.g., see \citeauthor{trainDiscreteChoiceMethods2003}, \citeyear{trainDiscreteChoiceMethods2003}).

In some circumstances, it might be unreasonable to assume that the budget set $(\mathbf{p}, y)$ is independent of the preference type $\eta$. When instruments are available, however, some forms of endogeneity can be handled by using a standard control function approach \citep{blundellEndogeneitySemiparametricBinary2004}. 

\subsection{Bounding the transition probabilities from the choice probabilities}\label{sec:set}
As mentioned before, the transition probabilities are nonparametrically identifiable and estimable from panel data that contains sufficient relative price and exogenous income variation. This immediately implies that all the results from previous subsections are also nonparametrically identified in such a data setting. One simply has to evaluate the estimated transition probabilities at virtual price vectors.

In many empirical applications, however, researchers only have access to (repeated) cross-sectional data. This type of data nonparametrically identifies the choice probabilities, but not the associated transition probabilities. However, by exploiting Boole-Fr\'echet \citep{frechet} and stochastic revealed preference inequalities, one can derive bounds on the now unobserved transition probabilities based on the observed choice probabilities.

\begin{propositionrep}
	\label{proposition_set}
		The transition probabilities $\{P_{i, j}(\mathbf{p}, \mathbf{p}', y)\}$ are set identified from the choice probabilities $\{P_i\}$ with bounds
		\begin{equation}\label{eq:bounds1}
		\begin{split}
		P^L_{i, i}(\mathbf{p}, \mathbf{p}', y) &= \max\left\{P_i(\mathbf{p},  y) + P_i(\mathbf{p}',  y) - 1, P_i\Big(\big(\max\{p_i, p'_i\}, \bmin\{\bp_{-i}, \bp'_{-i}\}\big), y\Big) \right\},\\ 
		P^U_{i, i}(\mathbf{p}, \mathbf{p}', y) &= 	\min \big\{ P_i(\mathbf{p},  y), \
		P_i(\mathbf{p}', y)\big\}.
		\end{split}
		\end{equation}
For $i\neq j$, $P_{i, j}(\mathbf{p}, \mathbf{p}', y) =0 $ if $p_i \geq p'_i$ and $p_{j} \leq p'_{j}$ and
		\begin{equation}\label{eq:bounds2}
		\begin{split}
		P^L_{i, j}(\mathbf{p}, \mathbf{p}', y) &= 
		    \max\{P_i(\mathbf{p},  y) + P_{j}(\mathbf{p}', y) - 1, \ 0\},  \\
		P^U_{i, j}(\mathbf{p}, \mathbf{p}', y) &= 
		    \min \left\{ P_i(\mathbf{p},  y), P_{j}(\mathbf{p}',y)\right\},
		\end{split}
		\end{equation}
		elsewhere. These bounds are sharp.
	\end{propositionrep}
\begin{proofEnd}
    		\label{proofproposition2}
	
	We first show that Equations~\eqref{eq:bounds1} and~\eqref{eq:bounds2} are valid bounds. One can immediately derive upper and lower bounds that are implied by elementary probability theory. Let $A$ be  the set $\{\eta |i = J^\eta(\bp,y) \}$ and $B$ the set $\{\eta |j = J^\eta(\bp',y) \}$. We have $P(A\cap B) = P_{i,j}(\mathbf{p},\mathbf{p}'; \ y)$, $P(A)=P_i(\mathbf{p}; \ y)$ and $P(B)=P_j(\mathbf{p}'; \ y)$.
	
	For the lower bound, note that 
	\begin{equation}
	    1 \geq P(A\cup B) = P(A) + P(B) - P(A\cap B)
	\end{equation}
	and hence $P(A\cap B) \geq P(A) + P(B) -1$ which translates into 
	\begin{equation}
	    P_{i,j}(\mathbf{p},\mathbf{p}'; \ y) \geq P_i(\mathbf{p}; \ y) + P_j(\mathbf{p}'; \ y) -1.
	\end{equation}
	For the upper bound, note that $P(A\cap B) \leq P(A)$ and $P(A\cap B) \leq P(B)$, and hence
	\begin{equation}
	    P_{i,j}(\mathbf{p},\mathbf{p}'; \ y)\leq \min(P_i(\mathbf{p}; \ y) , P_j(\mathbf{p}'; \ y)).
	\end{equation}
	These inequalities coincide with those derived by \cite{frechet}.
	
		We will now exploit the monotonicity condition imposed on the utility function $U^\eta_c$ to construct tighter bounds based on revealed preference restrictions. First consider the no-transition case. Note that if 
		\begin{equation}
		    U^\eta_i(y-\max\{p_i,p'_i\}) > U^\eta_k(y-\min\{p_k,p'_k\}),
		\end{equation}
	then $U^\eta_i(y-p_i) > U^\eta_k(y-p_k)$ and $U^\eta_i(y-p'_i) > U^\eta_k(y-p'_k)$ and hence 
	\begin{equation}
	   P_i\big((\max\{p_i,p'_i\},\bmin\{\bp_{-i},\bp'_{-i}\} ;y\big) = \Pr_\eta\Big[\cap_{k \neq i}\big\{U^\eta_i(y - \max\{p_i, p'_i\}) > U^\eta_k(y - \min\{p_k, p'_k\}) \big\}\Big] 
	\end{equation}
	is a lower bound of $P_{i,i}(\mathbf{p},\mathbf{p}'; \ y)$.
	
	Finally, for the transition case, some transitions are ruled out by monotonicity. Indeed, if $p_i \geq p'_i$ and $p_{j} \leq p'_{j}$, good $i$ becomes weakly less and good $j$ weakly more expensive after the price change. By monotonicity, it holds that $U^\eta_i(y-p_i) \leq U^\eta_i(y-p'_i)$ and $U^\eta_{j}(y-p_{j}) \geq U^\eta_{j}(y-p'_{j})$, and, hence, if moreover
	$ U^\eta_i(y - p_i) > U^\eta_k(y - p_k)  $ for all $k\neq i$ and $U^\eta_{j}(y - p'_{j}) > U^\eta_k(y - p'_k)$ for all $k\neq j$, then
	\begin{equation}
	   U^\eta_i(y-p'_i) \geq U^\eta_i(y-p_i)  > U^\eta_{j}(y - p_{j}) >  U^\eta_i(y - p'_i),
	\end{equation}
	which is a contradiction.
	Hence, if $p_i \geq p'_i$ and $p_{j} \leq p'_{j}$, then $P_{i,j}(\bp,\bp',y) = 0$.
	
	\bigskip
	We now demonstrate that the bounds derived above are sharp. To attain the lower bound implied by revealed preference restrictions, we construct a sequence of utility functions that maximizes the probability of ``staying" with choice $i$ when prices change. Specifically, consider $\bp$-dependent sequences $\{U^\eta_{c,b}: b = 1, 2, \dots\}$ such that, for alternative $i$,
	\begin{equation}
        \begin{cases}
            U^\eta_{i,b+1}(y- (p_i + \tau)) > U^\eta_{i,b}(y- (p_i + \tau)), & \text{ if } \tau >0, \\
	        U^\eta_{i,b+1}(y- (p_i + \tau)) = U^\eta_{i,b}(y- (p_i + \tau)), & \text{ if } \tau \leq 0, \\
        \end{cases}
	\end{equation}
	and 
    \begin{equation}
        \lim_{b \rightarrow \infty} U^\eta_{i,b}(y- (p_i + \tau)) = \lim_{b \rightarrow \infty} U^\eta_{i,b}(y- p_i), \quad \text{ if } \tau > 0.
    \end{equation}
    These conditions imply that individuals become unresponsive to price increases (i.e., $\tau > 0$) for alternative $i$ in the limit. Consequently, 
    \begin{equation}
        \lim_{b \rightarrow \infty} U^\eta_{i,b}(y- p_i) = \lim_{b \rightarrow \infty} U^\eta_{i,b}(y- \max\{p_i, p_i'\}).
    \end{equation}
    Similarly, for alternatives $c \neq i$, we require that  
    \begin{equation}
        \begin{cases}
	        U^\eta_{c,b+1}(y- (p_c + \tau)) = U^\eta_{c,b}(y- (p_c + \tau)), & \text{ if } \tau \geq 0, \\
            U^\eta_{c,b+1}(y- (p_c + \tau)) < U^\eta_{c,b}(y- (p_c + \tau)), & \text{ if } \tau <0, \\
        \end{cases}
	\end{equation}
    and 
    \begin{equation}
        \lim_{b \rightarrow \infty} U^\eta_{c,b}(y- (p_c + \tau)) = \lim_{b \rightarrow \infty} U^\eta_{c,b}(y- p_c), \quad \text{ if } \tau < 0.
    \end{equation}
    Thus, individuals become unresponsive to price decreases (i.e., $\tau < 0$) for alternatives $c \neq i$ in the limit. Consequently,
    \begin{equation}
        \lim_{b \rightarrow \infty} U^\eta_{c,b}(y- p_c) = \lim_{b \rightarrow \infty} U^\eta_{c,b}(y- \min\{p_c, p_c'\}), \quad \forall c \neq i.
    \end{equation}
From the definition of the transition probabilities in \eqref{eq:transprob}, we have that
	\begin{equation}
	    \begin{split}
	        \lim_{b\rightarrow \infty} P_{i, i}^b(\mathbf{p}, \mathbf{p}', y)&:= 
				\lim_{b\rightarrow \infty} \Pr_\eta \left[ \left \{U^\eta_{i,b}(y - p_i) \geq \max_{c \neq i}\{U^\eta_{c,b}(y-p_c)\}\right\} 
			\cap \left \{U^\eta_{i, b}(y - p'_i) \geq \max_{c \neq i} \{U^\eta_{c,b} (y-p'_c)\}\right\} \right] \\
			&= 
				\Pr_\eta \Bigg[ \lim_{b\rightarrow \infty} \left\{ U^\eta_{i,b}(y - \max\{p_i, p_i'\})  \geq  \max_{c\neq i} \left\{U^\eta_{c,b}(y-\min\{p_c, p_c'\})\right\}\right\} \\
                & \qquad \cap \lim_{b\rightarrow \infty} \left\{U^\eta_{i, b}(y - p'_i) \geq \max_{c \neq i} \{U^\eta_{c,b} (y-p'_c)\}\right\} \Bigg] \\
			 	&= 
				\Pr_\eta \left[\lim_{b\rightarrow \infty}  \left \{ U^\eta_{i,b}(y - \max\{p_i, p_i'\})  \geq  \max_{c\neq i} \left\{U^\eta_{c,b}(y-\min\{p_c, p_c'\})\right\}\right\} 
			 \right] \\
			 &= \lim_{b\rightarrow \infty} P_i^b\left(\max\{p_i,p'_i\},\bmin\{\bp_{-i},\bp'_{-i}\} ;y\right),
	    \end{split}
	\end{equation}
	where the second and fourth equalities follow because we consider a decreasing sequence of nested events.\footnote{Recall that for a decreasing  sequence of events $A_1 \supset A_2 \supset A_3 \supset \dots$ with limit $A = \cap_{m=1}^\infty A_m$, it holds that $\lim_{n \rightarrow \infty} \Pr[A_n] = \Pr[A]$.
	} The third equality holds because the first event is a proper subset of the second.
	
	Bivariate distributions that are on the  Boole-Fr\'echet bounds can be constructed by using insights from copula theory. Perfect positive dependence of the choice probabilities (i.e., \emph{co-monotonicity}) delivers the upper bound, while perfect negative dependence (i.e., \emph{counter-monotonicity}) delivers the lower bound. Consider an additive DC-RUM, i.e., $U^\eta_c(y-p_c) := V_c(y-p_c) + \zeta_c(\eta)$ for all alternatives $c$, for which we introduce the abbreviations $V_c = V_c(y-p_c)$ and $V'_c = V_c(y-p_c') $.
	
	Suppose that $\zeta_c(\eta) = 0$ for all $c$ except for a $k \neq i, j$. The transition probability $P_{i,j}(\bp, \bp', y)$ is then equal to
	\begin{equation}
	    \begin{split}
	        P_{i,j}(\bp, \bp', y) &= \Pr_\eta\left[\left\{V_i - V_k \geq \zeta_k(\eta)\right\} \cap \left\{V_j' - V_k' \geq \zeta_k(\eta)\right\}   \right] \mathbb{I}[V_i > V_c,\forall c\neq i,k]\mathbb{I}[V_j' > V_c', \forall c\neq j,k] \\
	        &=\Pr_\eta\left[\left\{\min\{V_i - V_k, V_j' - V_k'\} \geq \zeta_k(\eta)\right\} \right] \mathbb{I}[V_i > V_c,\forall c\neq i,k]\mathbb{I}[V_j' > V_c', \forall c\neq j,k]  \\
	        &= \min \left\{\Pr_\eta[V_i - V_k \geq \zeta_k(\eta)], \Pr_\eta[V_j' - V_k' \geq \zeta_k(\eta)]\right\} \mathbb{I}[V_i > V_c,\forall c\neq i,k]\mathbb{I}[V_j' > V_c', \forall c\neq j,k]\\
	        &= \min \left\{\Pr_\eta[V_i - V_k \geq \zeta_k(\eta)] \mathbb{I}[V_i > V_c,\forall c\neq i,k], \Pr_\eta[V_j' - V_k' \geq \zeta_k(\eta)] \mathbb{I}[V_j' > V_c', \forall c\neq j,k]\right\}\\
	        &= \min \{P_i(\bp, y), P_j(\bp', y)\},
	    \end{split}
	\end{equation}
	which is the Boole-Fr\'echet upper bound. 
	
	Now suppose that $\zeta_c(\eta) = 0$ for all $c$ except $i$, for which it is uniformly distributed on the unit interval, and suppose that $\mathbb{I}[V_j' > V_c', \forall c\neq j,i] = 1$ and that $0 \leq -\min_{c\neq i} \{V_i - V_c\} < V_j' - V_i' \leq 1$. In that case, the transition probability is equal to
	\begin{equation}
	    \begin{split}
	        P_{i,j}(\bp, \bp', y) &= \Pr_\eta\left[\left\{\min_{c\neq i}\{V_i - V_c\} \geq - \zeta_i(\eta)\right\} \cap \left\{V_j' - V_i' \geq \zeta_i(\eta)\right\}   \right] \mathbb{I}[V_j' > V_c', \forall c\neq j,i] \\
	        &=\Pr_\eta\left[\left\{-\min_{c\neq i} \{V_i - V_c\} \leq \zeta_i(\eta) \leq V_j' - V_i' \right\} \right]\mathbb{I}[V_j' > V_c', \forall c\neq j,i] \\
	        &= \left(\left(V_j' - V_i'\right) + \min_{c\neq i}\{V_i - V_c\}\right)\mathbb{I}[V_j' > V_c', \forall c\neq j,i] \\
	        &= P_j(\bp', y) + P_i(\bp, y) - 1,
	    \end{split}
	\end{equation}
	which is the Boole-Fr\'echet lower bound. 
\end{proofEnd}

The Boole-Fr\'echet inequalities ensure that the transition probabilities are weakly smaller than their associated marginal choice probabilities $P_i(\mathbf{p},  y)$ and $P_{j}(\mathbf{p}',  y)$. When $P_i(\mathbf{p},  y) + P_{j}(\mathbf{p}', y) - 1 >0$ they also deliver  nontrivial lower bounds. The stochastic revealed preference inequalities, which stem from the strong monotonicity of the utility function (see Assumption~\ref{assumption1}), provide additional identificational power in two particular instances. Firstly, by evaluating the choice probabilities at the least-favorable price vector $(\max\{p_i, p'_i\}, \bmin\{\bp_{-i}, \bp'_{-i}\})$, they yield an informative lower bound for the transition probabilities in the no-transition case where $i=j$. Secondly, when $i$ becomes weakly less expensive and $j \neq i$ becomes weakly more expensive, the transition probability should equal zero, as it is irrational for individuals to make this transition within the context of our model.

In the context of continuous choice, similar inequalities have been exploited by \cite{hoderleinstoye} and \cite{kitanurastoye}. The latter provide algorithmic tools to assess whether repeated cross-sectional data is compatible with the hypothesis of utility maximization. They do so by testing if there exist preference types that can rationalize the distribution of demand across segments (so-called \emph{patches}) of budget sets. While their computational approach could be employed to construct bounds on the transition probabilities in Equation~\eqref{eq:transprob}, Proposition~\ref{proposition_set} shows that there is a simple, closed-form solution in our setup. 

In an empirical application on labor supply \citep{longversion}, we find the identified sets for NOS measures to be small, which suggests that cross-sectional data is sufficiently rich for applied welfare analysis.

\subsection{Further practical guidance}\label{sec:implementation}

\paragraph{Differences in exogenous income.}
The assumption that the exogenous income $y$ is common to both situations with prices $\bp$ and $\bp'$ imposes no constraints on the transition probabilities $P_{i,j}(\bp, \bp', y)$. Indeed, if exogenous incomes are different when faced with prices $\bp$ and $\bp'$ (denoted by $y$ and $y'$, respectively), we can always redefine prices and incomes in order to obtain a common exogenous income. To see this, let $\bp'' = \bp'-y'+y$, such that
\begin{equation}
    \begin{split}
        P_{i,j}(\bp,\bp' ,y, y') &\equiv\Pr_\eta \left[ i=J^\eta(\bp,y), j = J^\eta(\bp',y')\right] \\
        &= 	\Pr_\eta \left[ i=J^\eta(\bp,y), j = J^\eta(\bp'-y' + y,y)\right]\\
        &= P_{i,j}(\bp,\bp'' ,y).
    \end{split}
\end{equation}

\paragraph{Presence of an outside option.}
Moreover, in some applications, there is an outside option that exhibits no independent price variation, which also hinders the direct empirical implementation of our approach. However, this difficulty can be circumvented by exploiting variation in the exogenous income $y$. Suppose alternative~$c_o \in \mathcal{C}$ is the outside option for which one has to evaluate the effect of a price change $\Delta p_o = p_o' - p_o$. By a change of variables, it then always holds that $P_i(\mathbf{p}', y) = P_i(\mathbf{p}'-\Delta p_o, y - \Delta p_o)$. Note that  the expression at the right-hand side does not require price variation for $c_o$, as $p_o' - \Delta p_o = p_o$ by construction.

\paragraph{Average welfare.}
A well-known implication of Fubini's theorem is that the mean of any random variable $X$, given that it exists, can be directly derived from its cumulative density function $F_X$, i.e.,
\begin{equation}
	\mathbb{E}_{F_X}(X) = \int_{0}^{\infty} (1 - F_X(u))du - \int_{-\infty}^{0}F_X(u)du.
\end{equation}
This result allows us to calculate average welfare from any of the distributional results derived in this paper. Note that when only bounds on the distribution of interest are available (see Section \ref{sec:set}), the expected value can be bounded by $\mathbb{E}_{F_X^U}(X) \leq \mathbb{E}_{F_X}(X) \leq \mathbb{E}_{F_X^L}(X)$, where $F^L_X$ and $F^U_X$ denote the CDF of the lower and upper bound respectively.

\section{Concluding remarks}\label{sec:conclusion}
In this paper, we provided a coherent framework to conduct individual and social welfare analysis for discrete choice. Allowing for unrestricted, unobserved preference heterogeneity, we argue that individual welfare measures become random variables from the point of view of the econometrician. For the class of NOS measures, we developed nonparametric methods to retrieve their distributions from observational data. In particular, we proved that all relevant marginal, conditional, and joint distributions can be expressed in terms of choice or transition probabilities, which are nonparametrically point-identified from cross-sectional and panel data, respectively. We also showed how transition probabilities can be set-identified when only cross-sectional data is available, which is important in empirical applications.

There are several promising directions for future research. Firstly, our results could be extended to settings where observed attributes of the alternatives, beyond prices, are subject to change. Along similar lines, future work could examine the welfare implications of introducing or removing certain alternatives. Such extensions are likely to result in set-identification rather than point-identification of the distributions of interest. Secondly, a valuable methodological advancement would be to incorporate measurement and optimization errors into the formal analysis. These errors may account for a substantial portion of the observed variation in choices, depending on the application, and could introduce bias into welfare estimates if not properly addressed. Lastly, it would be interesting to investigate the extent to which our results can be generalized to settings where some welfare-relevant outcomes are not choice variables or cannot be directly expressed in monetary terms. A canonical example of such an application is health.

\clearpage
\bibliography{references}

\clearpage

\setcounter{page}{1}
\begin{center}
    {\huge Appendix \\
    \vspace{10pt}
    \Large Identifying the Distribution of Welfare from Discrete Choice \\
    \normalsize \vspace{10pt}
    \textit{Bart Cap\'eau, Liebrecht De Sadeleer, and Sebastiaan Maes}
  }
\end{center}

\begin{appendix} \label{sec:onlineappendix}

\section{Distributional results for the EV}\label{app:add_results}
In this section, we derive analogue results to Theorem~\ref{thm:distribution_CV_initial_final}, Corollary~\ref{cor:CV}, and Theorem~\ref{thm:joint_distribution_CV_arbitrary}, but now for the EV instead of the CV. For an individual of type $\eta$, $EV^\eta$  is defined as 
\begin{equation}\label{eq:def_EV}
	\max_{c} \{U^\eta_c(y-p_c - EV^\eta)\} = \max_{c} \left\{U^\eta_c\left(y-p_c'\right)\right\},
\end{equation}
i.e, the amount of money (possibly negative) an individual has to pay before the reform to be equally well-off as after the reform.

\begin{thm}\label{thm:ev}
For the distribution of the EV, we have the following results:
\begin{align}
    \Pr_\eta[EV^\eta \leq z, i =J^\eta(\bp,y), j = J^\eta(\bp',y)] &= P_{i,j}(\bp, \bmin(\bp+z,\bp' ),y) \mathbb{I}\left[p_j' \leq p_j+ z\right], \\
    \Pr_\eta[EV^\eta \leq z \mid i =J^\eta(\bp,y), j = J^\eta(\bp',y)] &= \frac{P_{i,j}(\bp, \bmin(\bp+z,\bp' ),y)}{P_{i,j}(\bp, \bp', y)} \mathbb{I}\left[p_j' \leq p_j+ z\right], \\
    \Pr_\eta[EV^\eta \leq z \mid i = J^\eta(\bp,y)] &= \sum_{j}\frac{P_{i,j}(\bp, \bmin(\bp+z,\bp' ),y)}{P_i(\bp, y)}\mathbb{I}\left[p_j' \leq p_j+ z\right], \\
    \Pr_\eta[EV^\eta \leq z \mid j = J^\eta(\bp',y)] &= \frac{P_j(\bmin(\bp+z,\bp' ),y)}{P_j(\bp', y)} \mathbb{I}\left[p_j' \leq p_j+ z\right],\\
    \Pr_\eta[EV^\eta \leq z]  &= \sum_{j} P_j(\bmin(\bp+z,\bp' ),y)\mathbb{I}\left[p_j' \leq p_j+ z\right].
\end{align}{}
\end{thm}
\begin{proof}{}
    We have that 
    \begin{equation*}
	\begin{split}
		&\Big\{EV^\eta \leq z, i = J^\eta(\bp,y), j = J^\eta(\bp',y) \Big\} \\
		&=\Big\{\max_{c } \{U^\eta_c(y-p_c - EV^\eta)\} \geq \max_{c} \{U^\eta_c(y-p_c - z)\} , i = J^\eta(\bp,y), j = J^\eta(\bp',y) \Big\}\\
&= 	\Big\{\max_{c} \left\{U^\eta_c\left(y-p_c'\right)\right\} \geq \max_{c} \{U^\eta_c(y-p_c - z)\}, i = J^\eta(\bp,y), j = J^\eta(\bp',y) \Big\},\\
	 	\end{split}
\end{equation*} 
such that,
\begin{equation*}
	\begin{split}
	&\Pr_\eta[EV^\eta \leq z, i =J^\eta(\bp,y), j = J^\eta(\bp',y)]\\
	 &= \Pr_\eta\Big[U^\eta_i(y - p_i) \geq \max_{k\neq i}  U^\eta_k(y-p_k) ,\quad  U^\eta_j(y- p_j') \geq \max_{l\neq j}  U^\eta_l(y-p'_l) , \\
&	\qquad\qquad\qquad \max_c U^\eta_c(y-p_c') \geq \max_{c} \{ U^\eta_c(y-p_c - z)\}\Big] \\
&= \Pr_\eta\Big[U^\eta_j(y - p_j') \geq \max_{k\neq j}  U^\eta_k(y-\min(p_k+z,p_k')) ,\\
&\qquad\qquad U^\eta_i(y- p_i) \geq \max_{l\neq i}  U^\eta_l(y-p_l) \Big]  \mathbb{I}\left[p_j' \leq p_j + z\right]\\
&= P_{i,j}(\bp, \bmin(\bp+z,\bp' ),y) \mathbb{I}\left[p_j' \leq p_j+ z\right].
\end{split}\end{equation*}
The other equalities follow directly.
\end{proof}

\begin{thm}{}\label{thm:joint_initial_EV}
The joint distribution of the MMU, with initial prices as reference prices, and the EV is expressed as:
\begin{equation}\label{eq:joint_initial_EV}
\begin{split}
&\Pr_\eta[w \leq W^\eta_{M(\bp)}(y-p_i,i), EV^\eta \leq z , i = J^\eta(\bp,y), j = J^\eta(\bp',y)] \\
&\qquad\qquad =P_{i,j}\Big(\bp, \bmin(\bp', \bp+z),y\Big) \mathbb{I}\left[p_j' \leq p_j + z\right]\mathbb{I}\left[w \leq  y\right].\\
\end{split}
\end{equation}
\end{thm}

\begin{proof}
We have
\begin{equation*}
	\begin{split}
	    &\Pr_\eta[w \leq W^\eta_{M(\bp)}(y-p_i,i), EV^\eta \leq z, i = J^\eta(\bp,y), j = J^\eta(\bp',y)] \\
        &= \Pr_\eta\Big[U^\eta_i(y-p_i) \geq \max_{c'} U^\eta_{c'}(y-(p_{c'} + y -w)), \quad U^\eta_i(y - p_i) \geq \max_{k\neq i}  U^\eta_k(y-p_k) ,\\
            &	\qquad\qquad  U^\eta_j(y- p'_j) \geq \max_{l\neq j}  U^\eta_l(y-p'_l) , \quad U^\eta_j(y-p_j') \geq \max_{c } \{ U^\eta_c(y-p_c - z)\}\Big] \\
        &= \Pr_\eta\Big[U^\eta_i(y - p_i) \geq \max_{k\neq i}  U^\eta_k(y-\min(p_k, p_{k} + y - w)) ,\\
        &\qquad\qquad U^\eta_j(y- p_j') \geq \max_{l\neq j}  U^\eta_l(y-\min(p_l', p_l + z)) \Big]  \mathbb{I}\left[p_j' \leq p_j + z\right]\mathbb{I}\left[p_i \leq  p_i + y - w)\right]\\
        &= P_{i,j}\Big( \bp,(p_j', \bmin(\bp'_{-j}, \bp_{-j}+z)),y\Big) \mathbb{I}\left[p_j' \leq p_j + z\right]\mathbb{I}\left[w \leq  y\right]\\
        &= P_{i,j}\Big(\bp, \bmin(\bp', \bp+z),y\Big) \mathbb{I}\left[p_j' \leq p_j + z\right]\mathbb{I}\left[w \leq  y\right].
    \end{split}
\end{equation*}
\end{proof}

\clearpage
\section{Proofs}\label{app:proofs}
In this section, we provide the proofs of the results in the paper. Note that the results in Theorems~\ref{thm:distribution_CV_initial_final} and~\ref{thm:joint_distribution_CV_arbitrary}, and in Corollaries~\ref{cor:CV} and~\ref{cor:WandCV}, can in fact be seen as applications of Theorem \ref{thm:joint_distribution_general_arbitrary}. However, to provide more insight, we also give direct proofs below. 

\printProofs

\end{appendix}

\end{document}